\documentclass[sigconf,nonacm]{acmart} 
\AtBeginDocument{%
  \providecommand\BibTeX{{%
    \normalfont B\kern-0.5em{\scshape i\kern-0.25em b}\kern-0.8em\TeX}}}

\setcopyright{none}

\usepackage{balance} 
\usepackage{tikz}
\usepackage{tkz-graph}
\usepackage{tikz-cd}
\usetikzlibrary{tikzmark}
\usetikzlibrary{fadings}
\usetikzlibrary{patterns}
\usetikzlibrary{shadows.blur}
\usetikzlibrary{shapes}
\usetikzlibrary{positioning}
\usetikzlibrary{arrows}
\usetikzlibrary{arrows.meta}
\usetikzlibrary{bending}
\usetikzlibrary{topaths} 

\usepackage{algorithm}
\usepackage{algpseudocode}

\usepackage[titletoc]{appendix}

\algrenewcommand\algorithmicrequire{\textbf{Input:}}
\algrenewcommand\algorithmicensure{\textbf{Output:}}

\newcommand{\ef}{\textsf{EF}}
\newcommand{\sdef}{\textsf{SD-EF}}
\newcommand{\wef}{\textsf{WEF}}
\newcommand{\wsdef}{\textsf{WSD-EF}}

\newcommand{\prop}{\textsf{PROP}}
\newcommand{\sdprop}{\textsf{SD-PROP}}
\newcommand{\wprop}{\textsf{WPROP}}
\newcommand{\wsdprop}{\textsf{WSD-PROP}}

\newcommand{\efone}{\textsf{EF1}}
\newcommand{\sdefone}{\textsf{SD-EF1}}
\newcommand{\wefone}{\textsf{WEF1}}
\newcommand{\wsdefone}{\textsf{WSD-EF1}}

\newcommand{\propone}{\textsf{PROP1}}

\newcommand{\wpropone}{\textsf{WPROP1}}
\newcommand{\wsdpropone}{\textsf{WSD-PROP1}}

\newcommand{\wpropx}{\textsf{WPROPx}}
\newcommand{\wsdpropx}{\textsf{WSD-PROPx}}

\newcommand{\po}{\textsf{PO}}
\newcommand{\seq}{\textsf{SEQ}}

\begin{document}

\title[Proportional Allocations via Matchings]{Weighted Proportional Allocations of Indivisible Goods and Chores: Insights via Matchings}



\author{Vishwa Prakash H.V.}
\affiliation{
  \institution{Chennai Mathematical Institute}
  \city{Chennai}
  \country{India}}
\email{vishwa@cmi.ac.in}

\author{Prajakta Nimbhorkar}
\affiliation{
  \institution{Chennai Mathematical Institute}
  \city{Chennai}
  \country{India}}
\email{prajakta@cmi.ac.in}

\begin{abstract}
We study the fair allocation of indivisible goods and chores for agents with ordinal preferences and arbitrary entitlements. In both the cases of goods and chores, we show that there always exist allocations that are \textit{weighted necessarily proportional up to one item} (\wsdpropone{}), that is, allocations that are \wpropone{} under all additive valuations consistent with agents' ordinal preferences. We give a polynomial-time algorithm to find such allocations by reducing it to a problem of finding perfect matchings in a bipartite graph. We give a complete characterization of these allocations as extreme points of a perfect matching polytope. Using this polytope, we can optimize any linear objective function over all \wsdpropone{} allocations, for example, to find a min-cost \wsdpropone{} allocation of goods or most efficient \wsdpropone{} allocation of chores.  Additionally, we show the existence and computation of sequencible (\seq{}) \wsdpropone{} allocations by using rank-maximal perfect matching algorithms and show the incompatibility of Pareto optimality under all valuations and \wsdpropone{}. 

We also consider the Best-of-Both-Worlds (BoBW) fairness notion. By using our characterization, we give a polynomial-time algorithm to compute Ex-ante envy-free (\wsdef{}) and Ex-post \wsdpropone{} allocations for both goods and chores. 
\end{abstract}

\ccsdesc[500]{Theory of computation~Algorithmic game theory}
\ccsdesc[500]{Mathematics of computing~Matchings and factors}
\ccsdesc[300]{Mathematics of computing~Combinatorial algorithms}
\ccsdesc[300]{Theory of computation~Approximation algorithms analysis}


\keywords{Fair Division, Matchings, Proportionality}





\pagestyle{fancy}
\fancyhead{}


\maketitle 


\section{Introduction}\label{sec:intro}

Discrete fair allocation is a fundamental problem at the intersection of economics and computer science with applications in various multi-agent settings. Here, we are required to allocate a set of indivisible items to agents based on their preferences, such that each item is allocated to exactly one agent. This setting is commonly referred to as the \emph{assignment problem} \cite{gardenfors1973assignment,manlove2013algorithmics,foley1966resource,wilson1977assignment,sd}. In this setting, there is a set \(A\) of \(n\) agents and a set \(B\) of \(m\) indivisible items with each agent \(a_i\in A\) expressing an ordinal preference ordering over the items in $B$, given by a permutation \(\pi_i(B)\) of the items in \(B\). In addition to their ordinal preferences, agents may also have private cardinal valuations that reflect the utility or disutility of each item, ensuring compatibility with their ordinal preferences. The goal is to allocate the items to agents in a \emph{fair} manner. In this paper, we focus on {\em additive valuations}, where the utility (or disutility) of a set of items is the sum of the utilitites (or disutilities) of individual items.

The set \(B\) can represent \emph{goods}, covering scenarios such as inheritance division, house allocation, allocation of public goods, among others. Alternatively, \(B\) could be a set of \emph{chores}, modeling situations like task allocation among employees or household chore distribution between couples and so on.

Among various notions of fairness studied in the literature, two prominent ones are \emph{Envy-freeness} (\ef{}) and \emph{Proportionality} (\prop{}). An allocation is said to be {\em envy-free} if no agent would prefer to have the bundle held by any of the others. On the other hand, proportionality requires that each agent receives a set of items whose value is at least (at most, for chores) her proportional share of the total value of all the items. Unfortunately, \prop{} or \ef{} allocations do not always exist and are NP-hard to compute \cite{lipton2004approximately,Bouveret2014,sd}. Hence relaxations of these notions have been proposed in literature. \prop{} is relaxed as Proportionality \emph{up to one item} (\propone{}) \cite{prop1popublic,prop1market,prop1chores,propx-doesnt-exist} and  \ef{} is relaxed as Envy-free up to one item (\efone{})\cite{budish2011combinatorial,lipton2004approximately}.

In practical scenarios, agents can have varying \emph{entitlements} in situations such as inheritance division, division of shares among investors and so on. To capture such cases, a more generalized version of these notions, namely the weighted envy-freeness \wef{} \cite{chakraborty2021weighted,propx-doesnt-exist} and weighted proportionality \wprop{}\cite{bobw2023_1,bobw2023_2} are considered.

Given only the ordinal preferences, these notions of fairness are further strengthened by considering the \emph{stochastic dominance} (SD) relation. An agent prefers one allocation over another with respect to the SD relation if she gets at least as much utility from the former allocation as the latter for all cardinal utilities consistent with the ordinal preferences. An allocation is said to be weighted \emph{necessarily} proportional (also known as weighted strong SD proportionality) (\wsdprop{}) if it remains \wprop{} under all cardinal utilities consistent with the ordinal preferences. Similarly, the notion of envy-freeness (\ef{}) can be extended to weighted necessarily envy-freeness (\wsdef{}). Clearly, just like \prop{} and \ef{}, \wsdprop{} and \wsdef{} allocations may not exist. In fact, as shown in \cite{divorce,sd}, in an \sdprop{} allocation, each agent must receive their most favorite item - which is not realizable when two or more agents have the same most favorite item. For agents with varying entitlements and under ordinal preferences, these notions can be relaxed to \wsdefone{} and \wsdpropone{} (see e.g. \cite{wu2023weighted,chakraborty2021weighted,bobw2023_1}).

Given the non-existence of \wprop{} allocations, a well studied notion of fairness is that of \wpropx{}. Although \wpropx{} allocations exists under cardinal valuations \cite{li2022almost}, an analogous notion for the ordinal instances - namely, \wsdpropx{} allocations - need not exist (See Example~\ref{subsec:no_wpropx}). This further motivates the study of \wsdpropone{} allocations. 

Another approach to tackle the non-existence of \ef{} and \prop{} allocations is via randomization. A promising notion of fairness which has gained popularity over the recent years, is the notion of \emph{"Best of Both Worlds Guarantees"} (BoBW) \cite{bobw2020_1,bobw2020_2,bobw2021_1,bobw2023_1,bobw2023_2}. The aim is to compute a randomized fair allocation which also guarantees an approximate fairness notion in the deterministic setting. Suppose \(\mathcal{P}\) and \(\mathcal{Q}\) are two notions of fairness. Given a set of \(k\) allocations \(A=\langle A_1, A_2,\cdots,A_k \rangle\) and a probability distribution \(p=\langle p_1,p_2,\cdots,p_k \rangle\) over \(A\), the pair \(\langle A,p \rangle\) is said to be \emph{ex-post} \(\mathcal{Q}\) fair if each allocation \(A_1,A_2,\cdots,A_k\) are \(\mathcal{Q}\) fair and is called \emph{ex-ante} \(\mathcal{P}\) fair if \(\mathcal{P}\) fairness is guaranteed in expectation.

In \cite{bobw2020_1,bobw2020_2}, Aziz and Freeman et al. proposed polynomial time algorithms - the PS-Lottery algorithm (also called as the \textsc{Eating} algorithms) to compute ex-ante \sdef{} and ex-post \sdefone{} allocations of goods. In their approach, the agents are asked to hypothetically \emph{eat} the goods to produce a fractional allocation which is later decomposed into integral allocations. This approach was later modified in \cite{bobw2023_1} where the agents eat the goods at varying speeds to compute allocations that are ex-ante \wsdef{} and ex-post \wsdpropone{} along with \emph{weighted transfer envy-free up to one good} \textsf{WEF(1,1)}.

In practice, multiple allocations satisfying \wsdpropone{} might exist, and some could be {\em better} than others. For instance, different allocations of goods might incur different shipping/transportation costs or agents might have varying efficiency or expertise for each chore, independent of their own disutility or preference. In such cases, it becomes essential to optimize over the set of all \wsdpropone{} allocations. In this work, we particularly address this problem, providing a unified way to deal with goods and chores.

\subsection{Our Contributions}\label{subsec:contribution}
In this paper, we investigate fair allocation problems for agents with ordinal preferences and unequal entitlements. We provide the following key contributions:
\begin{itemize}
    \item In Theorem~\ref{thm:WPROP1-chores-exists} and Theorem~\ref{thm:WPROP1-goods-exists}, we show that the problem of existence and computation of \wsdpropone{} allocations reduces to that of the existence and computation of a perfect matching in a bipartite graph. We give such a reduction for both goods and chores. This gives a straightforward, matching based, polynomial-time algorithm that seamlessly adapts to computation of a \wsdpropone{} allocation of both goods and chores. 
    
    
    \item Moreover, we show that every perfect matching in the graph constructed above corresponds to a \wsdpropone{} allocation, and vice versa. Thus, we give a complete characterization of \wsdpropone{} allocations for both goods and chores as extreme points of a perfect matching polytope. This enables optimization of any linear objective function over the set of all \wsdpropone{} allocations. (See Section~\ref{sec:opt-chores})
    
    \item We study the economic efficiencies that can be guaranteed along with \wsdpropone{}. We provide a counter-example to show that Pareto optimality\footnote{Here, given only ordinal rankings, we call an allocation cardinally \po{} if it is \po{} under all cardinal valuations compatible with the ordinal rankings.} (\po{}) is not compatible with \wsdpropone{}. 
    \item On the positive side, In Theorem~\ref{thm:seq-chores} and Theorem~\ref{thm:seq-goods}, we show that every allocation that corresponds to a {\em rank-maximal} perfect matching in our perfect matchings instance is sequencible. En-route to this result, we show that rank-maximal matchings and rank-maximal perfect matchings in any bipartite graph are sequencible. See Lemma~\ref{lemma:rank-maximal-perfect}. This may be of independent interest.
    
    \item We also consider the best of both world fairness notion, in the context of our characterization. Our characterization leads to a simple polynomial-time algorithm for computing ex-ante \wsdef{} ex-post \wsdpropone{} allocations for both goods and chores. (See Theorem~\ref{thm:bobw-chores}). To the best of our knowledge, prior to this work, this result was not known for chores.
    
\end{itemize}
In the main part, we state the results for chores, and the results for goods are given in Appendix~\ref{sec:goods}.
{\em Extensions: } Our characterization of \wsdpropone{} allocations in terms of matchings paves a way to use tools from fairness in matchings to further generalization of the allocation problem. For instance, items (and also agents) can belong to various categories depending on their attributes, and there can be upper and lower quotas on each category of items that can be allotted and also on each category of agents as to the number of items they get. Existing results from the literature on fairness in matchings (e.g. \cite{IsraeliGapYear,IndiaReservation,PLN22,SankarLNN21}) can then be used to determine the existence of \wsdpropone{} allocations satisfying the category quotas, and outputting one if it exists.

\subsection{Related Work}\label{subsec:related_work}
Over the past two decades, there has been a growing interest in the study of computation of discrete fair division \cite{bezakova2005allocating,bouveret2010fair,lipton2004approximately, Bouveret2014, sd, aziz2022fair, nguyen2013survey, amanatidis2022fair}. In \cite{budish2011combinatorial,lipton2004approximately}, Budish et al and Lipton et al show the existence of \efone{} allocations for goods. In \cite{prop1popublic,prop1market}, Conitzer et al and Baman et al, extensively studied \propone{} allocations for goods. While significant advancements have been made in the allocation of goods, progress in the case of chores has been notably slower, and our understanding of chore allocation remains relatively limited in comparison to goods allocation. In 2019, Br{\^{a}}nzei and Sandomirskiy in \cite{prop1chores} extended the notion of \propone{} to the case of chores and gave a polynomial time\footnote{The algorithm runs in polynomial time when either the number of agents or the number of chores is fixed} algorithm to compute them. In 2020, Bhaskar et al showed the existance and polynomial time computation of \efone{} allocation of chores \cite{bhaskar2020approximate}. 

For agents with varying entitlements, In 2021, Chakraborty et al \cite{chakraborty2021weighted} showed the existence and  computation of \wefone{}+\po{} allocation in pseudo-polynomial time. They showed that in the context of weighted allocations, \wefone{} does not imply \wpropone{} for goods, in contrast to the unweighted case. Li et al, in \cite{li2022almost} showed the existence and computation of \wpropx{} (which implies \wpropone{}) allocation of chores. Later in 2023, Ex-post \wpropone{} allocations along with Ex-ante \wef{} allocations for goods were studied in \cite{bobw2023_1,bobw2023_2}. 
In \cite{bobw2023_2}, Aziz et al. proposed the Weighted Max Nash lottery Algorithm which computes an Ex-ante \po{} and Ex-post \wpropone{} allocation, for agents with additive cardinal valuations. 

In the weighted setting with agents expressing ordinal preferences for goods, Pruhs et al. in \cite{divorce}  reduced the problem of \wsdprop{} allocation of goods to that of finding perfect matchings in a bipartite graph. This was later generalized to preference lists with ties by Aziz et al. in \cite{sd}. While \wsdprop{} allocations may not always exist, their method provides a polynomial-time algorithm to compute one when it exists. Our reductions are inspired from their work. In 2023, Wu et al. proposed the \emph{Reversed Weighted Picking Sequence} Algorithm \cite{wu2023weighted} which always computes a \wsdefone{}+\seq{} (thus \wsdpropone{}+\seq{}) allocation of chores in polynomial time. 
\section{Preliminaries}\label{sec:prelim}

Let \(A=\{a_1,a_2,\ldots,a_n\}\) denote the set of \(n\) agents and let \(B=\{b_1,b_2,\ldots,b_m\}\) be set of \(m\) indivisible items. Each agent \(a_i\) expressing ordinal preferences over the items  given by a permutation \(\pi_i\) of the items in \(B\). The item set \(B\) can either be a set of goods or a set of chores. Each agent \(a_i\in A\) is endowed with an \emph{entitlement} \(\alpha_i \in [0,1]\) such that \(\sum_{a_i\in A} \alpha_i = 1\). 

Given an agent \(a_i\in A\), we denote the ordinal preference of \(a_i\) as a rank function \(\pi_i:[m]\to B\). The \(j^{\text{th}}\) rank item is given by \(\pi_i(j)\) and the rank of an item \(b\) is given by \(\pi_i^{-1}(b)\). In the case of goods, \(\pi_i(j)\) represents the \(j^{\text{th}}\)-most favorite good and in the case of chores, \(\pi_i(j)\) is the \(j^{\text{th}}\)-least favorite chore.

An instance of the allocation problem under ordinal valuations is represented by a tuple \(\mathcal{I} = \langle A, B, \Pi, \mathcal{F} \rangle\), where \(A\) and \(B\) are the sets of agents and goods, respectively. \(\Pi = \{\pi_1, \pi_2, \ldots, \pi_n\}\) denotes the set of rank functions, and \(\mathcal{F} = \{\alpha_1, \alpha_2, \ldots, \alpha_n\}\) represents the entitlements of each agent.

\paragraph{Fractional and Randomized Allocations: } We adopt the definitions of fractional and randomized allocations as outlined in \cite{bobw2020_1,bobw2020_2, bobw2021_1, bobw2023_2}. A \emph{fractional} allocation of the items in \(B\) to the agents in \(A\) is given by a non-negative \(n\times m\) matrix \(X=[x_{i,j}] \in [0,1]^{n\times m}\) such that an entry \(x_{i,j}\) denotes the fraction of the item  \(b_j\) allocated to the agent \(a_i\); for each item \(b_j\in B\), \(\sum_{a_i\in A} x_{i,j}=1\). We denote \(X_i=\langle x_{i,1},x_{i,2},\ldots,x_{i,m}\rangle\) as the \emph{fractional bundle} of items that is assigned to agent \(a_i\). A fractional allocation is integral, if \(x_{i,j}\in\{0,1\}\) for all \(a_i\in A\) and \(b_j \in B\). For an integral allocation \(X\), we denote with \textit{bundle} \(X_i\) the set of items that is assigned to agent \(a_i\) and the allocation \(X\) can be characterized by the bundles of the agents, i.e. \(X=\langle X_1,X_2,\ldots,X_n\rangle\).

A \textit{randomized allocation} is a lottery over integral allocations. In particular, a randomized allocation \(R\) is determined by \(k\) pairs \(\{(p_1,Y^1),(p_2,Y^2),\ldots,(p_k,Y^k)\}\), where each of the integral allocations \(Y^j\) for \(j\in[k]\), is implemented with probability \(p_j > 0\) and \(\sum_{j\in[k]}  p_j = 1\). We say that such an integral allocation is in the support of the randomized allocation. Moreover, we say that a fractional allocation \(X\) implements a randomized allocation \(Y\), if the marginal probability of agent \(a_i\) receiving item \(b_j\) is \(x_{i,j}\).

\paragraph{Cardinal Valuations: } Each agent \(a_i\in A\) can have a private cardinal valuation function \(v_i:2^B\to\mathbb{R}_{\ge0}\). When \(B\) represents a set of goods, \(v_i\) is called as a utility function. When \(B\) represents a set of chores, \(v_i\) is called as a disutility function. The heaviest (least favorite) chore is assigned the highest disutility value. We consider agent's valuation function to be additive, that is \(\forall S\subseteq B, v_i(S) = \sum_{b\in S} v_i(b)\).

A valuation function \(v_i\) is said to be \(\pi_i\)-\textit{respecting}, if \(v_i\) is consistent with the ordinal ranking \(\pi_i\). That is, \(\forall b,b'\in B,\  \pi^{-1}_i(b)<\pi^{-1}_i(b') \implies v_i(b)\ge v_i(b') \).  We denote the set of all \(\pi_i\)-respecting valuations as \(\mathscr{U}(\pi_i)\).

\paragraph{The Interval Representation of Items:} Consider an agent \(a_i\in A\) with an entitlement \(\alpha_i\).  We arrange the items along a number line from \(0\) to \(m\), such that the \(j^{\text{th}}\) rank item $\pi_i(j)$ occupies the interval $[j-1,j]$ for $1\leq j \leq m$. We refer to the interval $[j-1,j]$ as the item $\pi_i(j)$ itself. Furthermore, given an interval \(I=[p,q]\subseteq [0,m]\), we refer to \(I\)  as a fractional bundle itself.  If \(\delta\)  fraction of an interval \([j-1,j]\) overlaps with the interval \([p,q]\), that is \(|[j-1,j]\cap[p,q]|=\delta\) , then \(\delta\) fraction of the item \(b=\pi_i(j)\) belongs to the fractional bundle \(I\).  For a cardinal valuation function \(v_i\) of agent \(a_i\), the value of the bundle \(I=[p,q]\) is calculated as \(v_i(I)=\sum_{j\in[m]} |[j-1,j]\cap[p,q]|.v_i(\pi_i(j))\). 

Now, let the $[0,m]$ interval be sub-divided into \(k_i=\lceil m\alpha_i\rceil\) many intervals of lengths $\frac{1}{\alpha_i}$, except possibly the last interval, which can be shorter. The the $\ell^{\text{th}}$ interval is given by $I^i_\ell=\left[\frac{\ell-1}{\alpha_i},\frac{\ell}{\alpha_i}\right]$ for $1\leq \ell \leq \lfloor m\alpha_i\rfloor$, and, if $m\alpha_i$ is not integral, then the last interval is $\left[\frac{\lfloor m\alpha_i\rfloor}{\alpha_i},m\right]$. 
\begin{definition}
    For an agent \(a_i\in A\), we define the set of intervals \(I^i=\{I^i_1, I^i_2,\ldots,I^i_{k_i}\}\) as the \emph{interval set} of \(a_i\). 
\end{definition}
Note that if $\alpha_i<1$, length of each interval $\frac{1}{\alpha_i}>1$. Thus each interval \(I^i_\ell\), except possibly the last one, contains a non-zero portion of at least two consecutive items.

\paragraph{Stochastic Dominance(SD):}  A standard way of comparing fractional/randomized allocations is through first-order stochastic dominance. This notion has been extensively studied previously in \cite{BOGOMOLNAIA2001295,sd}.  An agent \(a_i\) prefers one allocation over another with respect to the SD relation if she gets at least as much value (or at most - in the case of chores) from the former allocation as the latter under all \(\pi_i\)-respecting cardinal valuations.

Suppose \(X_i\) and \(Y_i\) denote the fractional bundles of goods that an agent \(a_i\) receives  in the allocations \(X=[x_{ij}]\) and \(Y=[y_{ij}]\) respectively. We say that an agent \(a_i\) SD prefers \(X_i\) to \(Y_i\), denoted by \(X \succsim^{\text{SD}}_i Y\) if the following holds:
\[
    \forall j^*\in[m], \sum_{j:\pi_i(j)\ge \pi(j^*) } x_{i,j}\  \ge \sum_{j:\pi_i(j)\ge \pi(j^*) } y_{i,j}
\]
When \(X_i\) and \(Y_i\) denote bundles of chores, we say \(X \succsim^{\text{SD}}_i Y\) if the following holds:
\[
    \forall j^*\in[m], \sum_{j:\pi_i(j)\ge \pi(j^*) } x_{i,j}\  \le \sum_{j:\pi_i(j)\ge \pi(j^*) } y_{i,j}
\]

\subsection{Fairness and Efficiency Notions}

We begin with \emph{ex-ante - ex-post} notions as defined in \cite{BOGOMOLNAIA2001295,bobw2020_1,bobw2020_2,bobw2021_1,bobw2023_1,bobw2023_2}. For any property \(\langle{P}\rangle\) defined for an allocation, we say that a randomized allocation \(R\) satisfies \(\langle{P}\rangle\) \emph{ex-ante} if the allocation \(X\) that implements \(R\) satisfies \(\langle{P}\rangle\). For any property \(\langle{Q}\rangle\) defined for an integral allocation, we say that a randomized allocation \(R\) satisfies \(\langle{Q}\rangle\) \emph{ex-post} if every integral allocation in its support satisfies \(\langle{Q}\rangle\).

We now define various notions of weighted fairness under ordinal valuations. We start with the classic notion of \emph{envy-freeness.} Consider an instance of the allocation problem under ordinal valuations \(\mathcal{I} = \langle A, B, \Pi, \mathcal{F} \rangle\). 
\begin{definition}[\wsdef{}]\label{def:wsd-ef}
    Let \(B\) be a set of chores. An allocation \(X=\langle X_1,X_2,\ldots,X_n\rangle\) of \(B\) is said to be \emph{weighted} SD \emph{envy free (\wsdef{})}, if for every pair of agents \(a_i,a_k\in A\), we have
\[
    \frac{v_i(X_i)}{\alpha_i} \le \frac{v_i(X_k)}{\alpha_k} \quad\forall{v_i}\in \mathscr{U}(\pi_i),\forall{v_k}\in \mathscr{U}(\pi_k)
\]
And if \(B\) is a set of goods, then \(X\) is \emph{\wsdef{}} \cite{BOGOMOLNAIA2001295}, if for every pair of agents \(a_i,a_k\in A\), we have
\[
    \frac{v_i(X_i)}{\alpha_i} \ge \frac{v_i(X_k)}{\alpha_k} \quad\forall{v_i}\in \mathscr{U}(\pi_i),\forall{v_k}\in \mathscr{U}(\pi_k)
\]
\end{definition}

We consider the following notions of relaxed proportionality defined for integral allocations under cardinal valuations.
\begin{definition}[\wpropone{} \cite{propx-doesnt-exist}]\label{def:wprop1}
    Let \(B\) be a set of chores. In an integral allocation \(X=\langle X_1,X_2,\ldots,X_n\rangle\), a bundle \(X_i\) is said to be \emph{weighted}  \emph{proportional up to one item (\wpropone{})} for an agent \(a_i\) with a valuation function \(v_i\), if we have: 
\[
    \exists b\in X_i,\quad v_i(X_i\setminus \{b\}) \le \alpha_i.v_i(B)     
\]
And if \(B\) is a set of goods, then a bundle \(X_i\) is said to be \emph{\wpropone{}} for an agent \(a_i\) if we have: 
\[
    \exists b\in B,\quad v_i(X_i\cup \{b\}) \ge \alpha_i.v_i(B)     
\]
The allocation \(X\) is said to be \emph{\wpropone{}} if for all \(i\in[n]\), bundle \(X_i\) is \emph{\wpropone{}} for agent \(a_i\).
\end{definition}

Although the notion of \wpropone{} is conventionally defined for integral allocations, for the sake of analysis, we extend this notion to fractional allocations as follows:
\begin{definition}[fractional \wpropone{}]\label{def:wprop1f}
Let \(B\) be a set of chores. A fractional bundle \(X_i=\langle x_{i,1},x_{i,2},\ldots,x_{i,m}\rangle\) is \emph{\wpropone{}} for an agent \(a_i\) with a valuation function \(v_i\), if \( \exists b=\langle \beta_1,\beta_2,\ldots,\beta_m\rangle\), where \(\lVert{b}\rVert_1 = 1\) and \(0\le\beta_j\le x_{i,j}\) \(\forall j\in[m]\), we have \(v_i(X_i-b)\le \alpha_i.v_i(B)\). 

In the case of goods, a fractional bundle \(X_i\) is \emph{\wpropone{}} for an agent \(a_i\), if \( \exists b=\langle \beta_1,\beta_2,\ldots,\beta_m\rangle\), where \(\lVert{b}\rVert_1 = 1\) and \(\beta_j\ge0\) and \(x_{i,j}+\beta_j\le 1\) \(\forall j\in[m]\), we have \(v_i(X_i+b)\ge \alpha_i.v_i(B)\).
The allocation \(X\) is \emph{\wpropone{}} if bundle \(X_i\) is \emph{\wpropone{}} for every agent \(a_i\in A\).
\end{definition}

We can extend these definitions to the case of ordinal valuations as follows:
\begin{definition}[\wsdpropone{}]\label{def:wsd-prop1}
    An allocation \(X\) is said to be \emph{\wsdpropone{}}, if \(X\) is \emph{ \wpropone{} }for every agent \(a_i\in A\) under all valuations \(v_i \in \mathscr{U}(\pi_i)\)
\end{definition}

It is straightforward that for an agent \(a_i\), if a bundle \(X\) is \wsdpropone{}, then every bundle \(Y\) s.t. \(Y\succsim_i^{\text{SD}}X\) is also \wsdpropone{}.

Along with the notions of fairness, we study the following economic efficiencies considered in literature. 
\begin{definition}[Pareto Optimailty (\po{})]\label{def:po}
    For agents with cardinal valuations, an allocation \(X\) is said to be \emph{Pareto Optimal (\po{})} if there is no allocation \(Y\) that Pareto dominates it.  In the case of chore, this condition is expressed as  \(v_i(X_i) \le v_i(Y_i)\) for all \(i\in[n] \) and $\exists j\in [n]$ such that $v_j(X_j)<v_j(Y_j)$, while for goods, it is expressed as  \(v_i(X_i) \ge v_i(Y_i)\) for all \(i\in[n] \) and $\exists j\in [n]$ such that $v_j(X_j)>v_j(Y_j)$. 
\end{definition}

In the absence of Pareto Optimal allocations, a weaker notion of efficiency known as \textit{sequencibility} (\seq{})  is often considered. A \emph{picking sequence} of \(n\) agents for \(m\) items is an \(m\)-length sequence \(\sigma=\langle a'_1,a'_2,\ldots,a'_m\rangle\) where \(a'_i\in A\) for \(i\in[m]\).  An allocation \(X\) is the result of the picking sequence \(\sigma\) if it is the output of the following procedure: Initially every bundle is empty; then, at time step \(t\), agent \(a'_t\) inserts in her bundle the most preferred item among the available ones. Once an item is selected, it is removed from the set of the available items.
\begin{definition}[Sequencibility (\seq{})]\label{def:seq}
    An allocation \(X\) is said to be \emph{sequencible (\seq{})} if \(X\) is the result of some picking sequence \(\sigma\). 
\end{definition}
It is known that \po{}  implies \seq{}, and when number of agents \(n=2\), then \po{} is same as \seq{} \cite{Brams2002DividingTI}.
 
\subsection{Matchings}
Given a bipartite graph \(G=(A\cup B,E)\), an \(A\) perfect matching \(M\) is a matching in \(G\) that saturates all the vertices in \(A\). When \(|A|=|B|\), \(A\)-perfect matching is same as perfect matching. Given a matching \(M\) and a matched vertex \(a\in A\), we denote by \(M(a)\) the matched partner of \(a\). 

\paragraph{Rank-Maximal Matchings \emph{\cite{irving2003greedy, irving2006rank}}:} Consider a bipartite graph \(G=(A\cup B,E)\), s.t \(|A|=n,|B|=m\), where each vertex \(a\) in \(A\) ranks its neighbours \(N(a)\) from \(1\) to \(|N(a)|\). For each edge \((a,b)\in E\), let \ \(rank(a,b)\in [m]\) denote the rank of \(b\) in \(a\)'s ranking. The graph \(G\) along with the ranking is denoted as \(G=(A\cup B,E_1,E_2,\ldots,E_m )\) where \(E_i=\{(a,b)\in E\mid rank(a,b)=i\}\), for all \(i\in[n]\). A matching \(M\) in \(G\) can be decomposed as \(M=M_1\cup M_2\cup\cdots\cup M_m\) where \(M_i=M\cap{E_i}\). We define signature of a matching \(M\) in \(G\) as an \(m\) length tuple \(\rho(M)=\langle |M_1|,|M_2|,\ldots,|M_m| \rangle\).
\begin{definition}[Rank-Maximal Matching]
    Given a bipartite graph \(G=(A\cup B,E_1,E_2,\ldots,E_m)\), A matching \(M\) in \(G\) with lexicographically highest signature \(\rho(M)\) is called as a rank-maximal matching.
\end{definition}
Note that all rank-maximal matchings have identical signature. Furthermore, A rank-maximal matching need not be maximum size matching.
\begin{definition}[Rank-Maximal Perfect Matching]
    Given a bipartite graph \(G=(A\cup B,E_1,E_2,\ldots,E_m )\), a perfect matching \(M\) in \(G\) with lexicographically highest signature \(\rho(M)\) among all perfect matchings in \(G\) is called as a rank-maximal perfect matching.
\end{definition}

A matching \(M\) in \(G\) can be interpreted as an allocation of vertices in \(B\) to the vertices in \(A\). The ranks of the edges can be interpreted as the ordinal preferences of the vertices in \(A\). Under this interpretation, we borrow the definition of sequencibility (\seq{}) (Definition~\ref{def:seq}) for matchings. A matching \(M\) is said to be sequencible if the corresponding allocation \(M\) is sequencible. We denote by \(\rho(M)\), the picking sequence of vertices in \(A\) that constructs \(M\). 
\section{Existence and Computation of \wsdpropone{} Allocations for Chores via Matchings}\label{sec:existance-chores}

We now show that \wsdpropone{} allocations always exist for chores. To show this, we first characterize \wsdpropone{} bundles in Lemma~\ref{lemma:WPROP1-chores}. Using this lemma, we construct a bipartite graph \(G_c=(S\cup B,E)\) called an \emph{allocation graph} of a chores instance. We show that a \(B\)-perfect matching in \(G_c\) corresponds to a \wsdpropone{} allocation. We then use Hall's marriage condition \cite{Hall1987} to demonstrate that such a matching always exists, thus establishing the existence of \wsdpropone{} allocations. Later we extend these results to the case of goods in Appendix~\ref{sec:goods}
 
\begin{lemma}\label{lemma:WPROP1-chores} 
    Let \(T\subseteq B\) be a set of \(m_i\) chores, and let \(r_1<r_2<\cdots<r_{m_i}\) be the ranks of the chores in \(T\) in the ranking \(\pi_i\) of agent \(a_i\) (i.e, this set consists of the \(r_1\)-least favorite chore,  \(r_2\)-least favorite chore,\(\cdots\), and the \(r_{m_i}\)-least favorite chore for agent \(a_i\)). Then bundle \(T\) is \emph{\wsdpropone{}} for \(a_i\) if and only if the following two conditions hold: 

    \begin{eqnarray} 
        m_i & \le & \lfloor m\alpha_i\rfloor +1\label{eq:3}\\
    \forall 1\le \ell \le m_i,\quad
        r_\ell & \ge & \left \lceil\frac{\ell -1}{\alpha_i}\right \rceil \label{eq:4}
    \end{eqnarray}
\end{lemma}
\begin{proof}
Without loss of generality, for simplifying the notation, let the chores be renumbered according to the ranking of agent \(a_i\). Thus, \(b_j=\pi_i(j)\) for $1\leq j\leq m$. We  assume \(\alpha_i < 1\) as otherwise any bundle is \wsdpropone{} for agent \(a_i\) and further, \(\alpha_i >0\) as otherwise agent \(a_i\) can be removed from the instance. 

First, let us prove the necessity of these conditions. If any of the two conditions are not met, we exhibit a valuation \(v_i\) according to which, the bundle \(T\) is not \wpropone{} for agent \(a_i\). Suppose \(T\) violates condition~\ref{eq:3}. That is \(m_i\ge \lfloor m\alpha_i \rfloor +2\).  We set \(v_i(b_j)=1\) for all \(b_j\in B\). Under this valuation,
\[
    \forall b\in T,\quad v_i(T \setminus \{b\}) \ge \lfloor m\alpha_i\rfloor+1 > m\alpha_i = \alpha_i.v_i(B)
\]
Thus \(T\) is not a \wpropone{} bundle. Similarly, suppose \(T\) violates condition~\ref{eq:4}. That is,  \(r_\ell \le \left \lceil\frac{\ell -1}{\alpha_i}\right \rceil -1\) for some \(1\le \ell \le m_i\). We set \(v_i(b_j) = 1\) for all \(1\le j \le \left\lceil\frac{\ell -1}{\alpha_i}\right \rceil -1\) and \(v_i(b_j) = 0\) for all \(j\ge\left\lceil\frac{\ell -1}{\alpha_i}\right \rceil\). Under this valuation, \(\forall b\in T\) we have 
\[v_i(T\setminus \{b\}) \ge \ell-1 = \left(\frac{\ell-1}{\alpha_i}\right)\alpha_i>\left(\left\lceil\frac{\ell-1}{\alpha_i}\right\rceil-1\right)\alpha_i =\alpha_i.v_i(B)\]
Therefore, the bundle \(T\) is not \wpropone{}.

We now show the sufficiency of these conditions. Suppose conditions~\ref{eq:3} and \ref{eq:4} hold true for the bundle \(T\). It suffices to consider the case when both the conditions~\ref{eq:3} and \ref{eq:4} are tight, except \(r_1 = 1\).  This is because, for any other bundle \(Y_i=\{\pi(r'_1),\pi(r'_2)\cdots,\pi(r'_k)\}\) where \(1\le r'_1<r'_2<\cdots <r'_k\),  and at least one of the conditions \ref{eq:3} or \ref{eq:4} is not tight, we have \(Y_i \succsim^\text{SD}_i T\) since, for all \(1\le \ell \le k\),  \(r'_\ell \ge r_\ell\). 

To show that \(T\) is \wsdpropone{} for agent \(a_i\), we construct a fractional allocation \(T'\) using \(T\) such that \(T\succsim_i^{\text{SD}}T'\) and \(T'\) is a \wsdpropone{} allocation. 

Consider the interval set \(I^i=\{I^i_1, I^i_2,\ldots,I^i_{k_i}\}\)  of agent  \(a_i\). Include the chore \(b_1\) in \(T'\). For any other chore \(b_j\in T\setminus\{b_1\}\), we know that \(\pi^{-1}_i(b_j) = \left\lceil \frac{\ell -1}{\alpha_i} \right\rceil \)  for some \(\ell\in [m_i]\). Therefore, a non-zero fraction of the chore \(b_j\) lies in the right end of the interval \(I^i_{\ell-1} = \left[ \frac{\ell-2}{\alpha_i},\frac{\ell-1}{\alpha_i} \right]\). Suppose \(\frac{\ell-1}{\alpha_i} = \left\lceil \frac{\ell-1}{\alpha_i}\right\rceil -\delta\)  for some \(0<\delta\le1\). That is, \(1-\delta\) fraction of \(b_j\) lies in \(I^i_{\ell-1}\) and the remaining \(\delta\) portion lies in the interval \(I^i_\ell\). Then, from the interval \(I^i_{\ell-1}\) include the \(1-\delta\)  fraction of chore \(b_j\) and \(\delta\) fraction of the preceding chore \(b_{j-1}\) in \(T'\) (as shown in Figure~\ref{fig:5}). Under any valuation \(v_i \in \mathscr{U}{(\pi_i)} \),  for every chore \(b_j\) we have  \(\delta.{v_i(b_{j-1})}+(1-\delta)v_i(b_j)\ge v_i(b_j)\). Therefore it is clear that \(T\succsim^{\text{SD}}_{i}T'\). 
\begin{figure}
    \centering
 
\tikzset{
pattern size/.store in=\mcSize, 
pattern size = 5pt,
pattern thickness/.store in=\mcThickness, 
pattern thickness = 0.3pt,
pattern radius/.store in=\mcRadius, 
pattern radius = 1pt}
\makeatletter
\pgfutil@ifundefined{pgf@pattern@name@_igxdfejcy}{
\pgfdeclarepatternformonly[\mcThickness,\mcSize]{_igxdfejcy}
{\pgfqpoint{0pt}{0pt}}
{\pgfpoint{\mcSize+\mcThickness}{\mcSize+\mcThickness}}
{\pgfpoint{\mcSize}{\mcSize}}
{
\pgfsetcolor{\tikz@pattern@color}
\pgfsetlinewidth{\mcThickness}
\pgfpathmoveto{\pgfqpoint{0pt}{0pt}}
\pgfpathlineto{\pgfpoint{\mcSize+\mcThickness}{\mcSize+\mcThickness}}
\pgfusepath{stroke}
}}
\makeatother

 
\tikzset{
pattern size/.store in=\mcSize, 
pattern size = 5pt,
pattern thickness/.store in=\mcThickness, 
pattern thickness = 0.3pt,
pattern radius/.store in=\mcRadius, 
pattern radius = 1pt}
\makeatletter
\pgfutil@ifundefined{pgf@pattern@name@_gx32mk8fb}{
\pgfdeclarepatternformonly[\mcThickness,\mcSize]{_gx32mk8fb}
{\pgfqpoint{0pt}{0pt}}
{\pgfpoint{\mcSize+\mcThickness}{\mcSize+\mcThickness}}
{\pgfpoint{\mcSize}{\mcSize}}
{
\pgfsetcolor{\tikz@pattern@color}
\pgfsetlinewidth{\mcThickness}
\pgfpathmoveto{\pgfqpoint{0pt}{0pt}}
\pgfpathlineto{\pgfpoint{\mcSize+\mcThickness}{\mcSize+\mcThickness}}
\pgfusepath{stroke}
}}
\makeatother
\tikzset{every picture/.style={line width=0.75pt}} 

\begin{tikzpicture}[x=0.75pt,y=0.75pt,yscale=-1,xscale=1,scale=0.9]

\draw    (185,75) -- (503,75) ;
\draw [shift={(505,75)}, rotate = 180] [color={rgb, 255:red, 0; green, 0; blue, 0 }  ][line width=0.75]    (10.93,-3.29) .. controls (6.95,-1.4) and (3.31,-0.3) .. (0,0) .. controls (3.31,0.3) and (6.95,1.4) .. (10.93,3.29)   ;
\draw    (185,75) -- (495,75) (215,73) -- (215,77)(245,73) -- (245,77)(275,73) -- (275,77)(305,73) -- (305,77)(335,73) -- (335,77)(365,73) -- (365,77)(395,73) -- (395,77)(425,73) -- (425,77)(455,73) -- (455,77)(485,73) -- (485,77) ;
\draw    (185,75) -- (185,7) ;
\draw [shift={(185,5)}, rotate = 90] [color={rgb, 255:red, 0; green, 0; blue, 0 }  ][line width=0.75]    (10.93,-3.29) .. controls (6.95,-1.4) and (3.31,-0.3) .. (0,0) .. controls (3.31,0.3) and (6.95,1.4) .. (10.93,3.29)   ;
\draw  [fill={rgb, 255:red, 155; green, 155; blue, 155 }  ,fill opacity=1 ] (185,25) -- (215,25) -- (215,75) -- (185,75) -- cycle ;
\draw   (185,86) .. controls (185,90.67) and (187.33,93) .. (192,93) -- (225,93) .. controls (231.67,93) and (235,95.33) .. (235,100) .. controls (235,95.33) and (238.33,93) .. (245,93)(242,93) -- (278,93) .. controls (282.67,93) and (285,90.67) .. (285,86) ;
\draw   (285,86) .. controls (285,90.67) and (287.33,93) .. (292,93) -- (325,93) .. controls (331.67,93) and (335,95.33) .. (335,100) .. controls (335,95.33) and (338.33,93) .. (345,93)(342,93) -- (378,93) .. controls (382.67,93) and (385,90.67) .. (385,86) ;
\draw  [pattern=_igxdfejcy,pattern size=6pt,pattern thickness=0.75pt,pattern radius=0pt, pattern color={rgb, 255:red, 0; green, 0; blue, 0}] (280,45) -- (305,45) -- (305,75) -- (280,75) -- cycle ;
\draw  [pattern=_gx32mk8fb,pattern size=6pt,pattern thickness=0.75pt,pattern radius=0pt, pattern color={rgb, 255:red, 0; green, 0; blue, 0}] (385,50) -- (395,50) -- (395,75) -- (385,75) -- cycle ;
\draw   (424,87) .. controls (424.07,91.67) and (426.44,93.96) .. (431.11,93.88) -- (444.62,93.66) .. controls (451.28,93.55) and (454.65,95.83) .. (454.73,100.5) .. controls (454.65,95.83) and (457.94,93.44) .. (464.61,93.33)(461.61,93.38) -- (478.12,93.11) .. controls (482.78,93.04) and (485.07,90.67) .. (485,86) ;
\draw  [fill={rgb, 255:red, 155; green, 155; blue, 155 }  ,fill opacity=1 ] (275,45) -- (285,45) -- (285,75) -- (275,75) -- cycle ;
\draw  [fill={rgb, 255:red, 155; green, 155; blue, 155 }  ,fill opacity=1 ] (255,35) -- (275,35) -- (275,75) -- (255,75) -- cycle ;
\draw  [fill={rgb, 255:red, 155; green, 155; blue, 155 }  ,fill opacity=1 ] (355,50) -- (385,50) -- (385,75) -- (355,75) -- cycle ;
\draw   (245,35) -- (255,35) -- (255,75) -- (245,75) -- cycle ;
\draw   (335,50) -- (365,50) -- (365,75) -- (335,75) -- cycle ;
\draw    (295,35) .. controls (290.08,-3.42) and (276.42,-3.99) .. (265.5,28.49) ;
\draw [shift={(265,30)}, rotate = 287.93] [color={rgb, 255:red, 0; green, 0; blue, 0 }  ][line width=0.75]    (10.93,-3.29) .. controls (6.95,-1.4) and (3.31,-0.3) .. (0,0) .. controls (3.31,0.3) and (6.95,1.4) .. (10.93,3.29)   ;
\draw    (390,45) .. controls (383.14,14.62) and (371.48,19.77) .. (360.66,43.52) ;
\draw [shift={(360,45)}, rotate = 293.75] [color={rgb, 255:red, 0; green, 0; blue, 0 }  ][line width=0.75]    (10.93,-3.29) .. controls (6.95,-1.4) and (3.31,-0.3) .. (0,0) .. controls (3.31,0.3) and (6.95,1.4) .. (10.93,3.29)   ;
\draw  [fill={rgb, 255:red, 155; green, 155; blue, 155 }  ,fill opacity=1 ] (455,65) -- (485,65) -- (485,75) -- (455,75) -- cycle ;

\draw (182,79.4) node [anchor=north west][inner sep=0.75pt]  [font=\tiny]  {$0$};
\draw (212,79.4) node [anchor=north west][inner sep=0.75pt]  [font=\tiny]  {$1$};
\draw (242,79.4) node [anchor=north west][inner sep=0.75pt]  [font=\tiny]  {$2$};
\draw (272,79.4) node [anchor=north west][inner sep=0.75pt]  [font=\tiny]  {$3$};
\draw (302,79.4) node [anchor=north west][inner sep=0.75pt]  [font=\tiny]  {$4$};
\draw (332,79.4) node [anchor=north west][inner sep=0.75pt]  [font=\tiny]  {$5$};
\draw (362,79.4) node [anchor=north west][inner sep=0.75pt]  [font=\tiny]  {$6$};
\draw (392,79.4) node [anchor=north west][inner sep=0.75pt]  [font=\tiny]  {$7$};
\draw (403,76.4) node [anchor=north west][inner sep=0.75pt]  [font=\footnotesize]  {$\cdots $};
\draw (482,78.4) node [anchor=north west][inner sep=0.75pt]  [font=\tiny]  {$m$};
\draw (174,64) node [anchor=north west][inner sep=0.75pt]  [font=\scriptsize,rotate=-270] [align=left] {disutility};
\draw (441,77) node [anchor=north west][inner sep=0.75pt]  [font=\scriptsize] [align=left] {chores};
\draw (231,101.4) node [anchor=north west][inner sep=0.75pt]  [font=\scriptsize]  {$I_{1}^{i}$};
\draw (331,100.4) node [anchor=north west][inner sep=0.75pt]  [font=\scriptsize]  {$I_{1}^{i}$};
\draw (451,100.4) node [anchor=north west][inner sep=0.75pt]  [font=\scriptsize]  {$I_{k_{i}}^{i}$};

\end{tikzpicture}
    \caption{Construction of the fractional allocation \(T'\) (shaded in grey) that is SD dominated by \(T\). }
    \label{fig:5}
\end{figure}
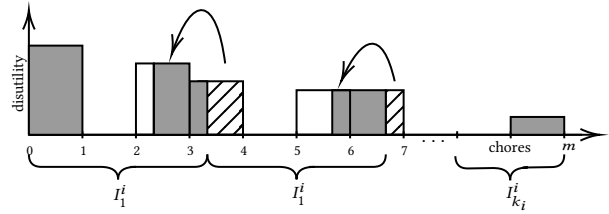

From the construction of \(T'\), we know \(T'\setminus \{b_1\}\) contains the least valued one unit of chore from each interval (except possibly the last interval which could have no contribution to \(T'\)). Therefore, \(v_i(T-b_1) \le \sum_{j=1}^{k_i} \alpha_i v_i(I^i_j) = \alpha_i.v_i(B)\). Thus, \(T'\) is a \wsdpropone{} bundle. 
\end{proof}

With the help of this characterization, we now construct an \textit{allocation graph} \(G_c\) of chores. Given a fair allocation instance \(\mathcal{I}=\langle A, B, \Pi, \mathcal{F} \rangle\), we construct a bipartite graph - the allocation graph \(G_c=(S\cup B,E)\) as follows:

\renewcommand{\labelitemi}{-}
\begin{itemize}
    \item The set of chores \(B\) forms one part of \(G_c\) with chores interpreted as vertices. 
    \item For every agent \(a_i\in A\), and every \(\ell = 1,2,\cdots, m_i=\lfloor m\alpha_i\rfloor +1\), create vertex \(s_{i,\ell}\) in \(S\). We call these the \(m_i\) many \emph{slots} of agent \(a_i\).
    \item From each slot \(s_{i,\ell}\), draw edges to every chore \(b\) for which \( \pi_i^{-1}(b) \ge \left \lceil\frac{\ell -1}{\alpha_i}\right \rceil \). That is, \((s_{i,\ell},b) \in E \iff \pi_i^{-1}(b) \ge \left \lceil\frac{\ell -1}{\alpha_i}\right \rceil  \) 
\end{itemize}

The allocation graph \(G_c\) exhibits several interesting properties:

Firstly, we have 
\begin{proposition}\label{perfect-matchings-wsdprop1}
\(B\)-perfect matching in \(G_c\), i.e a matching that saturates all the chores, satisfies Conditions~\ref{eq:3} and \ref{eq:4} and this corresponds to a \wsdpropone{} allocation of chores. Conversely, any \wsdpropone{} allocation satisfies Conditions~\ref{eq:3} and \ref{eq:4} and thus forms a \(B\)-perfect matching in \(G_c\).
\end{proposition}
Moreover, in the interval set \(I^i=\{I^i_1, I^i_2,\ldots,I^i_{k_i}\}\) of an agent \(a_i\), if non-zero fraction of a chore \(b\) lies in the interval \(I^i_\ell = \left[ \frac{\ell-1}{\alpha_i},\frac{\ell}{\alpha_i} \right]\), then \(\pi_i^{-1}(b)\ge \left \lceil\frac{\ell -1}{\alpha_i}\right \rceil\). Thus slot \(s_{i,\ell}\) has an edge to chore \(b\) in \(G_c\). This is formally stated in the following proposition.

\begin{proposition}\label{prop:slot-interval}
    In the allocation graph \(G_c\) of a chore allocation instance \(\mathcal{I}\), each slot \(s_{i,\ell}\) of each agent \(a_i\in A\), has edges to every chore with a non-zero portion in the interval \(I^i_\ell\) in the interval set of \(a_i\).
\end{proposition}

We also observe the following property regarding edge relationships in \(G_c\):
\begin{proposition}\label{prop:heavy_chores}
     Let \(S_i = \{s_{i,\ell} \mid \ell \in [m_i]\}\) represent the set of slots in \(S\) that belong to an agent \(a_i\in A\). For any two chores \(b_j\) and \(b_k\) in \(B\), such that \(\pi_i^{-1}(b_k) \ge \pi_i^{-1}(b_j)\), the neighbourhood of \(b_k\) in \(S_i\) contains the neighbourhood of \(b_j\) in \(S_i\). That is, \(N(b_j)\cap S_i \subseteq N(b_k)\cap S_i \). Therefore, for any set of chores \(Y=\{b_1,b_2,\cdots,b_k\}\), the following holds:
    \[
         S_i \cap N\left(\{\pi_i(1),\pi_i(2),\cdots,\pi_i(k)\}\right) \subseteq S_i \cap  N(Y)
    \]
\end{proposition}

We now show the existence of \wsdpropone{} allocations by showing that a \(B\)-perfect matching always exists in \(G_c\). To prove this, we rely on the Hall's marriage theorem \cite{Hall1987} which characterizes the existence of perfect matchings in bipartite graphs.

\begin{theorem}[Hall's Theorem \cite{Hall1987}]\label{thm:halls}
    Given a bipartite graph \(G=(A\cup B,E)\), there exists an \(A\)-perfect matching in \(G\) if and only if for every subset \(S\subseteq A\), the number of vertices in the neighbourhood \(N(A)\) is greater than or equal to the size of \(S\):
    \[\forall S\subseteq A,\ |N(S)| \ge |S|\]
\end{theorem}

Now, we establish the following main result:

\begin{theorem}\label{thm:WPROP1-chores-exists}
    For any fair allocation instance of chores \(\mathcal{I}=\langle A,B,\Pi,\mathcal{F} \rangle\), there always exists a \wsdpropone{} allocation.
\end{theorem}
\begin{proof}
    Consider the allocation graph \(G_c\) of \(\mathcal{I}\). We show that \(G_c\) always has a \(B\)-perfect matching and thus \(\mathcal{I}\) has a \wsdpropone{} allocation. Let \(T=\{b_1,b_2,\cdots,b_k\}\subseteq B\) be set of \(k\) vertices in \(B\). The goal is to show that the size of the neighbourhood \(N(T)\) is bigger than or equal to the size of \(T\). From Proposition~\ref{prop:heavy_chores}, we can assume \emph{w.l.o.g} that for each agent \(a_i\in A\), the chores \(b_1,b_2,\cdots,b_k\) are the first (lowest rank) \(k\) chores, since it minimizes the neighbourhood. For an agent \(a_i\), let \(s_{i,\ell_i}\) be the highest index slot which has an edge to \(b_k\). Thus, all the slots \(s_{i,1},s_{i,2},\cdots,s_{i,\ell_i}\) have edges to \(b_k\). Therefore, the size of the neighbourhood \(|N(T)| \ge \sum_{i\in [n]} \ell_i\). Since the slot \(s_{i,\ell_i+1}\) does not have an edge to \(b_k\), Condition~\ref{eq:4} is violated. Thus,
    \begin{align*}
                & k < \left \lceil\frac{(\ell_i +1) -1}{\alpha_i}\right \rceil\\
        \implies& k < \frac{\ell_i}{\alpha_i}+1\\
        \implies& (k-1)\alpha_i < \ell_i \\
        \implies& \sum_{i \in [n]}(k-1)\alpha_i < \sum_{i \in [n]}\ell_i &\text{(Summing over all agents)}\\
        \implies& k-1 < |N(T)| &\text{(We know \(\sum_{i \in [n]}\ell_i \le |N(T)|\))}\\
        \implies& k \le |N(T)| &\text{(As both \(|N(T)|\) and \(k\) are integers)}
    \end{align*}
    Therefore, For any set of \(k\) chores the size of the neighbourhood is more than or equal to \(k\). From Theorem~\ref{thm:halls}, the allocation graph \(G_c\) always has a \(B\)-perfect matching - which corresponds to a \wsdpropone{} allocation of chores. 
\end{proof}
Therefore, using the famous Hopcroft-Karp algorithm \cite{Hopcroft-Karp} to find perfect matchings, we can compute a \wsdpropone{} allocation in time \(\mathcal{O}(m+n)^{2.5}\)
\section{Optimizing Over Allocations}\label{sec:opt-chores}

Recall that any \(B\)-perfect matching in the allocation graph \(G_c\) corresponds to a \wsdpropone{} allocation and vice versa. In this section we extend the allocation graph \(G_c\) to \(G^+_c\) by balancing the two parts of the bipartite graph while maintaining the correspondence between \wsdpropone{} allocations and perfect matchings in \(G^+_c\). We then optimize different linear objective functions over all \wsdpropone{} allocations using the perfect matching polytope.

\paragraph{Extending the Allocation Graph:} Consider the allocation graph \(G_c=(S\cup B,E)\) of an instance \(\mathcal{I}\) of chores allocation. For each agent \(a_i\in A\), there are \(m_i=\lfloor m\alpha_i \rfloor + 1\) many slots in \(S\). Therefore the total number of slots \(|S| = n + \sum_{i\in [n] }\lfloor m\alpha_i \rfloor \). To construct \(G^+_c=(S\cup B',E')\), we create \(|S|-m = q\) many additional \emph{dummy} chores \(b'_1,b'_2,\cdots,b'_q\) in \(B'\) to  balance the bipartite graph. Draw additional edges from all the slots in \(S\) to every dummy chore. 

A \wsdpropone{} allocation of chores gives a \(B\)-perfect matching in \(G_c\). We can extend this matching to a perfect matching in \(G^+_c\) by matching the dummy chores in any manner as all the dummy chores have edges to every slot. Conversely, given a perfect matching in \(G^+_c\), we can ignore the edges from dummy chores to get a \(B\)-perfect matching in \(G_c\) and thus a \wsdpropone{} allocation. 

Given a bipartite graph \(G=(X\cup Y,E)\), the following constraints define the matching polytope:
\begin{align}
\label{eqn:LP}
\begin{split}
    \sum_{x\in N(y)}e_{xy} &= 1 \quad \forall y\in Y\\
    \sum_{y\in N(x)}e_{xy} &= 1 \quad \forall x\in X\\
    e_{xy} &\geq 0
\end{split}
\end{align}
We know that above matching polytope is integral \cite{lovaszplummar} and hence a matching that maximizes a given objective function is computable in polynomial-time \cite{KHACHIYAN198053,Karmarkar}. We now use this fact to compute \wsdpropone{} allocations while considering agents' \textit{efficiency} in doing the chores.

\subsection{Considering Agent Competence }\label{subsec:fair-efficient-chore}
Regardless of how each agent personally values any given chore, it is important to acknowledge that their skills and proficiency in performing them can vary significantly across different tasks. For any specific agent-chore pairing \(a_i,b\), we can quantify the agent's competence in performing chore \(b\) as \(u_i(b)\in[0,1]\), where \(0\) indicates low competency and \(1\) indicates high competency. This efficiency metric helps us assess how well-suited each agent is to tackle a particular chore, guiding us in achieving a fair and efficient chore allocation. 

We can use the above given linear program formulation to maximize efficiency over all \wsdpropone{} allocations. Given an extended allocation graph \(G^+_c=(S,B',E')\), We set the objective function as follows:
 \begin{align*}
     \text{maximize}\quad \sum_{(s_{i,\ell},b)\in E'} u_i(b)\cdot e_{(i,\ell),b}
 \end{align*}
Similarly, we can optimize for time spent on doing chores and other linear objective functions. 
\section{Best of Both Worlds}\label{sec:bobw}
In this section, using the perfect matchings in the extended allocation graph \(G^+_c\), we give a polynomial time algorithm to compute an ex-ante \wsdef{} and ex-post \wsdpropone{} allocation of chores. 

We begin by constructing a \wsdef{} allocation \(X=\langle X_1,X_2,\cdots,X_n\rangle\). For each agent \(a_i\in A\), give \(\alpha_i\) fraction of every chore \(b\in B\). In this allocation, For each pair of agents \(a_i\) and \(a_k\), for any two valuations \(v_i\in \mathscr{U}(\pi_i), v_k\in \mathscr{U}(\pi_k)\), we know that \(\frac{v_i(X_i)}{\alpha_i} = \frac{v_i(X_k)}{\alpha_k} = v_i(B)\) and hence \(X\) is a \wsdef{} allocation. We now show that this fractional allocation can be realized as a fractional perfect matching in the extended allocation graph \(G^+_c\).

\begin{lemma}\label{lemma:fractionalmatching}
    Given an instance \(\mathcal{I}=\langle A,B,\Pi,\mathcal{F} \rangle\) of chore allocation, there exists a fractional perfect matching in the extended allocation graph \(G^+_c=(S\cup{B'},E')\) of \(\mathcal{I}\) that corresponds to a \wsdef{} chore allocation where each agent \(a_i\) receives \(\alpha_i\) fraction of every chore.
\end{lemma}
\begin{proof}
    We first construct a fractional matching \(M\) that saturates all the \emph{real} chores (non dummy chores). Such a matching can always be extended to a fractional perfect matching by assigning the dummy chores in any manner, as all the dummy chores have edges to all the slots. 
    
    Consider the interval set \(I^i\) of an agent \(a_i\in A\). From Proposition~\ref{prop:slot-interval}, we know that slot \(s_{i,\ell}\) has edges to every chore in the interval \(I^i_\ell\). With the help of this fact, we construct a fractional matching \(M\) in \(G^+_c\) as follows:
    
    Let \(x_{i,\ell,b}\) denote the fraction of the edge \((s_{i,\ell},b)\) in \(M\). let \(\delta_{b,\ell}\) denote the fraction of a chore \(b\in B\) that is present in the interval \(I^i_\ell\). For every edge \((s_{i,\ell},b)\), we set \(x_{i,\ell,b}=\alpha_i\cdot\delta_{b,\ell}\). A slot \(s_{i,\ell}\) receives non zero fractions of the chores from the interval \(I^i_\ell\). Each slot receives at most \(1\) unit of chore because total chores assigned for a slot \(s_{i,\ell}\) is : 
    \begin{align*}
    \sum_{b\in B}x_{i,\ell,b} = \alpha_i\sum_{b\in B}\delta_{b,\ell} \le \alpha_i\frac{1}{\alpha_i} = 1
    \end{align*}
    The fraction of a given real chore \(b\) received by agent \(i\) across all the slots is:
    \begin{align*}
        \sum_{\ell=1}^{m_i} x_{i,\ell,b} =  \alpha_i\sum_{\ell=1}^{m_i}\delta_{b,\ell} = \alpha_i
    \end{align*}

    The matching \(M\) saturates all the real chores. Since the graph \(G^+_c\) is a balanced bipartite graph, and as all the dummy chores have edges to all the slots, the matching \(M\) can be extended to a fractional perfect matching by dividing the dummy chores across the remaining spaces of all the slots in any arbitrary way. 
\end{proof}
Let us denote this fractional perfect matching as \(M^*\). Note that \(M^*\) lies inside the matching polytope of \(G^+_c\). We now decompose this fractional perfect matching into convex combination of integral perfect matchings with the help of Birkhoff’s decomposition. Given a perfect matching \(M\)(fractional or otherwise) of a balanced bipartite graph \(G=(P\cup Q,E)\) with \(2n\) vertices, \(M\) can be represented as a \(n\times n\) bi-stochastic matrix \(X=(x_{ij})\) where an entry \(x_{ij}\) denotes the fraction of the edge \((i,j)\) present in \(M\). Given a fractional perfect matching, we can decompose it as a convex combination of integral perfect matchings with the help of Birkhoff-von-Neumann theorem \cite{Birkhoff,vonNeumann,Johnson1960OnAA,lovaszplummar}. 
\begin{theorem}[Birkhoff-von Neumann]\label{thm:bvn}
    Let \(X\) be a \(p \times p\) bi-stochastic matrix. There exists an algorithm that runs in \(\mathcal{O}(p^{4.5})\) time and computes a decomposition \(X=\sum_{k=1}^q \lambda_k X_k\) where \(q\le p^2 -p +2\); for each \(k\in[q],\lambda_k\in[0,1]\) and \(X_k\) is a \(p\times p\) permutation matrix; and \(\sum_{k=1}^q \lambda_k = 1\).
\end{theorem}
With the help of Theorem~\ref{thm:bvn}, we design Algorithm~\ref{alg:BoBW-chores}: The Uniform Lottery Algorithm, which gives an ex-ante \wsdef{} and ex-post \wsdpropone{} allocation of chores using only the ordinal valuations.

\begin{algorithm}
\caption{\textsc{Uniform Lottery Algorithm} for chores}\label{alg:BoBW-chores}
\begin{algorithmic}[1]
    \Require A chore allocation instance \(\mathcal{I}=\langle{A},B,\Pi,\mathcal{F}\rangle\), where \(|A|=n\) and \(|B|=m\).
    \Ensure A fractional \wsdef{} allocation \(X=\sum_{k=1}^q~\lambda_k~X_k\) where each \(X_k\) represents a deterministic \wsdpropone{} allocation and \(q\in \mathcal{O}(m^c)\).
    \State \(G^+_c \gets\) extended allocation graph of \(\mathcal{I}\)
    \State \(Y \gets \) fractional perfect matching in \(G^+_c\) where each agent \(a_i\in A\) gets \(\alpha_i\) fraction of every real chore \Comment{(As in Lemma~\ref{lemma:fractionalmatching})}
    \State Invoke Theorem~\ref{thm:bvn} to compute a decomposition \(Y=\sum_{k=1}^q~\lambda_k Y_k\) where \(q\le(m+n)^2-(m+n)-2\)
    \State Convert \(Y=\sum_{k=1}^q\lambda_k{Y_k}\) to \(X=\sum_{k=1}^q\lambda_k{X_k}\) where all the dummy chores are ignored.\\
    \Return Allocation \(X\) and its decomposition \(\sum_{k=1}^q\lambda_k{X_k}\)
\end{algorithmic}
\end{algorithm}
\begin{theorem}\label{thm:bobw-chores}
    The randomized allocation implemented by Algorithm~\ref{alg:BoBW-chores} is ex-ante \wsdef{} and ex-post \wsdpropone{}
\end{theorem}
\begin{proof}
    Algorithm~\ref{alg:BoBW-chores} returns an allocation \(X\) and its decomposition \(\sum_{k=1}^q\lambda_k{X_k}\). From Lemma~\ref{lemma:fractionalmatching}, we know that the allocation \(X\) returned by the algorithm is \wsdef{}. Each of the \(X_k\)s in the decomposition is a \(B\)-perfect matching in the allocation graph \(G_c\).Therefore, from Proposition~\ref{perfect-matchings-wsdprop1}, each \(X_k\) is \wsdpropone{}.
\end{proof}
\section{Beyond Fairness: Economic Guarantees}\label{sec:seq-chore}

In Section~\ref{sec:existance-chores}, we discussed the reduction from \wsdpropone{} allocations to matchings. In this section, we investigate the incorporation of additional economic efficiency notions alongside fairness. In Appendix Section~\ref{subsec:impossibility-chores}, we give an example instance where no \wsdpropone{} allocation is Pareto optimal under all valuations. Furthermore, we give examples of instances where given the cardinal valuations, there does not exist a Pareto optimal allocation that is \wsdpropone{} for the underlying ordinal instance. Therefore, we explore a more relaxed concept known as \emph{sequencibility} (\seq{}). We prove a general graph theoretic lemma, showing that every rank-maximizing \(A\)-perfect matching is sequencible. This lemma could be of independent interest with other applications. Using this result, we establish that computing a rank-maximizing perfect matching, rather than an arbitrary one, yields \wsdpropone{}+\seq{} allocations.

We begin with the following simple observation about rank-maximal matchings:
\begin{proposition}\label{prop:seq}
    Given a graph \(G=(A\cup B,E=E_1\cup\ldots\cup E_r)\), if \(M\) is a rank-maximal matching in \(G\), then $M$ is sequencible (\seq{}). 
\end{proposition}
\begin{proof}
Define $G_i=(A\cup B, E_1\cup\ldots\cup E_i)$. Let the rank-maximal matching \(M\) be decomposed as $M=M_1\cup M_2\cup \ldots M_r$ where $M_i=M\cap E_i$ for $i\in [r]$. For a vertex $a$, we denote its matched partner in $M$ as $M(a)$.
A picking sequence for $M$ is $\langle A_1, A_2,\ldots,A_r\rangle$ where, for $i\in[r]$, $A_i$ is an arbitrary ordering of vertices in $A$ that are matched by an edge in $M_i$. This sequence results in $M$, for the following reason. Firstly, $M_1$ is a maximum matching in $G_1$. For each $i\in [r], i>1$, after the agents in $A_1\cup\ldots\cup A_{i-1}$ pick the items they are matched to in $M$, $M_i$ is a maximum matching on rank $i$ edges in the remaining graph. Thus, when all the agents in $A_1\cup\ldots\cup A_i$ pick their favorite item among the available items, there are no items left, which are ranked between $1$ and $i$ for any of the remaining agents, this results in agents in $A_{i+1}$ picking their rank $i+1$ items.
\end{proof}
\balance
Our interest is in finding an \(A\)-perfect matching that is also sequencible. In general, a rank-maximal matching need not be an \(A\)-perfect matching and all perfect matchings are not sequencible.The Figure~\ref{fig:rank-maximal} shows one such example. We show that rank-maximal perfect matchings are sequencible. Unlike a rank-maximal matching, a rank-maximal perfect matching $M$ may not satisfy the properties mentioned in Proposition~\ref{prop:seq} i.e., $M_1$ may not be a maximum matching in $G_1$, and in general, after agents in $A_1\cup\ldots\cup A_{i-1}$ pick their respective choices, $M_i$ may not be a maximum matching on rank $i$ edges in the remaining graph. Hence the ordering of vertices in $A_1,\ldots,A_r$ needs to be carefully chosen while constructing the picking sequence.
\begin{figure}
    \centering
    \input{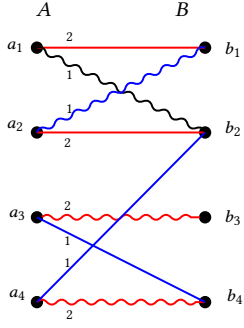}
    \caption{The edges in red form a perfect matching - but it is not sequencible. The blue edges correspond to a rank-maximal matching, but it is not perfect. The squiggly edges corresponds to a rank-maximal \(A\)-perfect matching. This is sequencible, and the picking sequence is $\langle a_1,a_2,a_4,a_3\rangle$}
    \label{fig:rank-maximal}
\end{figure}
\begin{lemma}\label{lemma:rank-maximal-perfect}
    Given a graph \(G=(A\cup B,E_1\cup\ldots\cup E_r)\) which is an instance of the rank-maximal matchings problem, and an \(A\)-perfect matching \(M\) in \(G\), if \(M\) is a rank-maximal \(A\)-perfect matching then \(M\) is sequencible. 
\end{lemma}
\begin{proof}
Consider a rank-maximal \(A\)-perfect matching \(M\). Let \(A_i=\{a\in A\mid rank(a,M(a))=i\}\) for all \(i\in[r]\). We show that there always exists an ordering $\sigma_i(A_i)$ of vertices within each $A_i$ such that $\sigma(A)=\langle \sigma_1(A_1),\ldots,\sigma_r(A_r)\rangle $ is a picking sequence for \(M\). Clearly, each vertex $a\in A_1$ can appear in any order in $\sigma(A_1)$. As long as it appears before vertices in $A_i, i>1$, it gets to pick its first choice, which is $M(a)$. Since $M_1$ need not be maximum in $G_1$, a vertex in $A_2$ may still have its first choice left unmatched after vertices in $A_1$ choose their match. However, to get $M$, we need all the vertices in $A_i$ to pick their $i$\textsuperscript{th} choice. We show the existence of a picking sequence by induction on the length of the partial sequence constructed at any point. Consider the stage where vertices in $A'\subseteq A$ are already arranged in a picking sequence, and $A_1,\ldots,A_{i-1}\subseteq A'$, and we are currently constructing $\sigma_i$. We show that there is always a vertex $a\in A_i\setminus A'$ that has its first $i-1$ choices already matched to vertices in $A'$. Thus, $a$ can be the next vertex in the sequence, and it will have to pick its $i$\textsuperscript{th} choice i.e. $M(a)$. Define $B'=M(A')$. Define $A''=A\setminus A'$, $B''=B\setminus B'$.

Call an edge \((a,b)\in E\setminus M\), with $a\in A'', b\in B''$, a {\em $+$-edge} if $rank(a,b)<rank(a,M(a))$ i.e. $a$ prefers $b$ over its match in $M$. In the picking sequence, we need to choose a vertex $a\in A''$ that has no $+$-edge to any $b\in B''$. Assume for the sake of contradiction that such an $a$ does not exist. Then consider the graph $G''=(A''\cup B'', E'')$, where $E''$ is set of edges $(a,b)$ where $a\in A''$, $b\in B''$ and $(a,b)\in M$ or $(a,b)$ is a $+$-edge in $G$. Since $M$ is an \(A\)-perfect matching, the degree of every vertex in $A''$ is at least $2$. Thus $G''$ must either have an alternating path $P$ or a cycle $C$ that alternates between $+$-edges and edges of $M$. Now, the symmetric difference $M\oplus C$ (or $M\oplus P$) is a matching that is also perfect and has a better signature than $M$, since every agent appearing on $C$ gets a more preferred choice in $M\oplus C$ by the definition of a $+$-edge. This contradicts the assumption that $M$ is a rank-maximal perfect matching. Thus, there must be an $a\in A''$ with no $+$-edge incident on it, and hence $M(a)$ is its most preferred choice in $B''$. This vertex can be inserted as the next vertex in the picking sequence, breaking ties arbitrarily, if any.
\end{proof} 

Using Lemma~\ref{lemma:rank-maximal-perfect}, we now show that a rank-maximal perfect matching in the extended allocation graph \(G^+_c\) gives a sequencible \wsdpropone{} allocation.

\begin{theorem}\label{thm:seq-chores}
    There always exist a \wsdpropone{}+\seq{} allocation of chores.
\end{theorem}
\begin{proof}
Let \(G^+_c=(S,B',E')\) be the extended allocation graph of an instance \(\mathcal{I}\). Recall that \(B\) is the set of real chores and \(B'\setminus B\) is the set of dummy chores and \(|B|=m,|B'|=m+q\). For each slot \(s_{i,\ell}\), we first rank the real chores from \(1\) to \(m\) as \(rank(s_{i,\ell},b)=m+1-\pi_i(b)\) for all \(b\in B\). The dummy chores ranked from \(m+1\) to \(m+q\) in an arbitrary way.

For any two slots \(s_{i,p}\) and \(s_{i,q}\) of an agent \(a_i\), if \(p>q\), then \(N(p)\subseteq N(q)\). This is because \(G^+_c\) satisfies Condition~\ref{eq:4}. Therefore, given a matching \(M\), if \(rank(s_{i,p},M(s_{i,p})>rank(s_{i,q},M(s_{i,q}))\), then we can interchange \(M(s_{i,p})\) and \(M(s_{i,q})\) without altering the signature of the matching. Thus, given a rank-maximal perfect matching \(M\), we can assume w.o.l.g that for any agent \(a_i\in A\), and \(p,q\le m_i\), if \(p>q\) then \(rank(s_{i,p},M(s_{i,p})<rank(s_{i,q},M(s_{i,q}))\).

Given a rank-maximal perfect matching \(M\) in \(G^+_c\), from Lemma~\ref{lemma:rank-maximal-perfect} we obtain a sequence \(\sigma(S)\) of slots. To construct a sequence of agents, replace each \(s_{i,\ell}\) with the corresponding agent \(a_i\). Since dummy chores are ranked higher than real chores, all the slots that are matched to dummy chore forms the tail of the sequence \(\sigma(M)\) and hence they can be safely ignored. Therefore, a rank-maximal perfect matching in \(G^+_c\) gives a \wsdpropone{}+\seq{} allocation. 
\end{proof}

Therefore, using the algorithm to find rank-maximal perfect matchings \cite{michail2007reducing,irving2003greedy}, we can compute a \wsdpropone{}+\seq{} allocation in time \(\mathcal{O}((m+n)^{3.5})\).
\section{Conclusion}\label{sec:conclusion}
In this paper, we consider the fairness notion of weighted necessarily proportionality up to one item (\wsdpropone{}). We show that finding \wsdpropone{} allocations can be reduced to finding perfect matchings in a bipartite graph - namely the allocation graph. This insight provides a practical framework for leveraging tools and techniques from the field of matching theory. We show that rank-maximal perfect matchings give picking sequences for finding \wsdpropone{}+\seq{} allocations. We show that the perfect matching polytope of the allocation graph captures all the \wsdpropone{} allocations, thus enabling us to optimize for any linear objective function. We then create a fractional perfect matching in the allocation graph, corresponding to a \wsdef{} allocation. Decomposing this allocation, equivalent to decomposing the fractional matching into integral matchings, results in a randomized algorithm for computing an Ex-ante \wsdef{} Ex-post \wsdpropone{} allocation, both in the case of goods and chores. Our works raises the open question of the existence of \wsdpropone{} allocations in the mixed setting, where set \(B\) includes both goods and chores.
\begin{acks}

We express gratitude to the anonymous reviewers of AAMAS 2024 for their valuable comments. H.V. acknowledges support from TCS-RSP for their grant. Special thanks to friends and colleagues at CMI - Archit Chauhan, Asif Khan, Soumodev Mal, and Sheikh Shakil Akhtar for their time and expertise in proofreading and discussing the contents of this paper.

\end{acks}
\clearpage

\bibliographystyle{ACM-Reference-Format}
\bibliography{main}

\newpage
\begin{appendices}
\appendixpage
\section{Allocating Goods}\label{sec:goods}
We now extend all our results to the case of goods. Given an instance \(\mathcal{I}=\langle A,B,\Pi,\mathcal{F} \rangle\) of goods allocation, we first characterize \wsdpropone{} bundles as follows: 

\begin{lemma}\label{lemma:WPROP1-goods} 
    Let \(T\subseteq B\) be a set of \(m_i\) goods, and let \(r_1<r_2<\cdots<r_{m_i}\) be the ranks of the goods in \(T\) in the ranking \(\pi_i\) of agent \(a_i\) (i.e, this set consists of the \(r_1\)-most favorite good,  \(r_2\)-most favorite good,\(\cdots\), and the \(r_{m_i}\)-most favorite good for agent \(a_i\)). Then bundle \(T\) is \emph{\wsdpropone{}} for \(a_i\) if and only if the following two conditions hold: 

    \begin{eqnarray} 
        m_i & \ge & \lceil m\alpha_i\rceil -1\label{eq:1}\\
    \forall 1\le \ell \le m_i,\quad
        r_\ell & \le & \left \lfloor\frac{\ell}{\alpha_i}\right \rfloor +1 \label{eq:2}
    \end{eqnarray}
\end{lemma}
\begin{proof}
Without loss of generality, for simplifying the notation, let the goods be renumbered according to the ranking of agent \(a_i\). Thus, \(b_j=\pi_i(j)\) for $1\leq j\leq m$. We  assume \(\alpha_i < 1\) as otherwise, allocating all goods to \(a_i\) is a \wsdpropone{} allocation. Further, \(\alpha_i >0\) as otherwise agent \(a_i\) can be removed from the instance. 

First, let us prove the necessity of these conditions. If any of the two conditions are not met, we exhibit a valuation \(v_i\) according to which, the bundle \(T\) is not \wpropone{} for agent \(a_i\). Suppose \(T\) violates condition~\ref{eq:1}. That is \(m_i\le \lceil m\alpha_i \rceil -2\).  We set \(v_i(b_j)=1\) for all \(b_j\in B\). Under this valuation,
\[
    \forall b\in B,\quad v_i(T \cup \{b\}) \le \lceil m\alpha_i\rceil-1 < m\alpha_i = \alpha_i.v_i(B)
\]
Thus \(T\) is not a \wpropone{} bundle. Similarly, suppose \(T\) violates condition~\ref{eq:2}. That is,  \(r_\ell\ge\left\lfloor\frac{\ell}{\alpha_i}\right\rfloor+2\) for some \(1\le \ell \le m_i\). We set \(v_i(b_j) = 1\) for all \(1\le j\le\left\lfloor\frac{\ell}{\alpha_i}\right\rfloor+1\) and \(v_i(b_j) = 0\) for all the remaining goods. Under this valuation,
\begin{align*}
    \forall b\in B,\quad v_i(T\cup\{b\}) &= \ell \\
                                               &= \alpha_i\left(\frac{\ell}{\alpha_i}\right)\\
                                               &< \alpha_i\left(\left\lfloor\frac{\ell}{\alpha_i}\right\rfloor+1\right) =\alpha_i. v_i(B)
\end{align*} Therefore, the bundle \(T\) is not \wpropone{}.

We now show the sufficiency of these conditions. Suppose conditions~\ref{eq:1} and \ref{eq:2}. hold true for the bundle \(T\). It suffices to consider the case when both the conditions~\ref{eq:1} and \ref{eq:2} are tight. This is because for any other bundle \(Y_i=\{\pi(r'_1),\pi(r'_2)\cdots,\pi(r'_k)\}\) where \(1\le r'_1<r'_2<\cdots r'_k\),  and at least one of the conditions \ref{eq:1} or \ref{eq:2} is not tight, we have \(Y_i \succsim^\text{SD}_i T\) since, for all \(1\le \ell \le k\),  \(r'_\ell \le r_\ell\). 

To show that \(T\) is \wsdpropone{} for agent \(a_i\), we construct a fractional allocation \(T'\) using \(T\) such that \(T\succsim_i^{\text{SD}}T'\) and \(T'\) is a \wsdpropone{} allocation. 

Consider the interval set \(I^i=\{I^i_1,I^i_2,\ldots,I^i_{k_i}\}\)  of agent  \(a_i\). For each good \(b_j\in T\), we know that \(\pi^{-1}_i(b_j) = \left\lfloor \frac{\ell}{\alpha_i} \right\rfloor +1 \)  for some \(\ell\in [m_i]\). For every interval excluding the first interval, i.e., \(I_{\ell+1}=\left[\frac{\ell}{\alpha_i},\frac{\ell+1}{\alpha_i}\right]\) for all \(1< \ell \le k_i-1\), a non-zero fraction of the good \(\pi^{-1}_i(r_{\ell})\), which occupies the interval \(\left[\left\lfloor\frac{\ell}{\alpha_i}\right\rfloor,\left\lfloor\frac{\ell}{\alpha_i}\right\rfloor+1\right]\) lies in the beginning the interval \(I_{\ell+1}\). Suppose \(\left\lfloor\frac{\ell}{\alpha_i}\right\rfloor+1 = \frac{\ell}{\alpha_i} +\delta\)  for some \(0\le\delta<1\). That is, \(1-\delta\) fraction of \(b_j\) lies in \(I^i_{\ell-1}\) and the remaining \(\delta\) portion lies in the interval \(I^i_\ell\). Then, from the interval \(I^i_{\ell}\) include the \(\delta\)  fraction of good \(b_j\) and \(1-\delta\) fraction of the subsequent good \(b_{j+1}\) in \(T'\) (as shown in Figure~\ref{fig:6}). Under any valuation \(v_i \in \mathscr{U}{(\pi_i)} \),  for every good \(b_j\) we have \(\delta.{v_i(b_{j})}+(1-\delta)v_i(b_{j+1})\le v_i(b_j)\). Therefore it is clear that \(T\succsim^{\text{SD}}_{i}T'\). 

From the construction of \(T'\), we know that \(T'\cup \{b_1\}\) contains the highest valued one unit of good from each interval. Therefore, \(v_i(T \cup \{b_1\}) \ge \sum_{j=1}^{k_i} \alpha_i v_i(I^i_j) = \alpha_i.v_i(B)\). Thus, \(T'\) is a \wsdpropone{} bundle. 

\begin{figure}[h]
    \centering
 
\tikzset{
pattern size/.store in=\mcSize, 
pattern size = 5pt,
pattern thickness/.store in=\mcThickness, 
pattern thickness = 0.3pt,
pattern radius/.store in=\mcRadius, 
pattern radius = 1pt}
\makeatletter
\pgfutil@ifundefined{pgf@pattern@name@_0up0rcxv1}{
\pgfdeclarepatternformonly[\mcThickness,\mcSize]{_0up0rcxv1}
{\pgfqpoint{0pt}{0pt}}
{\pgfpoint{\mcSize+\mcThickness}{\mcSize+\mcThickness}}
{\pgfpoint{\mcSize}{\mcSize}}
{
\pgfsetcolor{\tikz@pattern@color}
\pgfsetlinewidth{\mcThickness}
\pgfpathmoveto{\pgfqpoint{0pt}{0pt}}
\pgfpathlineto{\pgfpoint{\mcSize+\mcThickness}{\mcSize+\mcThickness}}
\pgfusepath{stroke}
}}
\makeatother

 
\tikzset{
pattern size/.store in=\mcSize, 
pattern size = 5pt,
pattern thickness/.store in=\mcThickness, 
pattern thickness = 0.3pt,
pattern radius/.store in=\mcRadius, 
pattern radius = 1pt}
\makeatletter
\pgfutil@ifundefined{pgf@pattern@name@_bbr4bjnyd}{
\pgfdeclarepatternformonly[\mcThickness,\mcSize]{_bbr4bjnyd}
{\pgfqpoint{0pt}{0pt}}
{\pgfpoint{\mcSize+\mcThickness}{\mcSize+\mcThickness}}
{\pgfpoint{\mcSize}{\mcSize}}
{
\pgfsetcolor{\tikz@pattern@color}
\pgfsetlinewidth{\mcThickness}
\pgfpathmoveto{\pgfqpoint{0pt}{0pt}}
\pgfpathlineto{\pgfpoint{\mcSize+\mcThickness}{\mcSize+\mcThickness}}
\pgfusepath{stroke}
}}
\makeatother
\tikzset{every picture/.style={line width=0.75pt}} 

\begin{tikzpicture}[x=0.75pt,y=0.75pt,yscale=-1,xscale=1,scale=0.9]

\draw    (223,75) -- (528,75) ;
\draw [shift={(530,75)}, rotate = 180] [color={rgb, 255:red, 0; green, 0; blue, 0 }  ][line width=0.75]    (10.93,-3.29) .. controls (6.95,-1.4) and (3.31,-0.3) .. (0,0) .. controls (3.31,0.3) and (6.95,1.4) .. (10.93,3.29)   ;
\draw    (185,75) -- (520,75) (215,73) -- (215,77)(245,73) -- (245,77)(275,73) -- (275,77)(305,73) -- (305,77)(335,73) -- (335,77)(365,73) -- (365,77)(395,73) -- (395,77)(425,73) -- (425,77)(455,73) -- (455,77)(485,73) -- (485,77)(515,73) -- (515,77) ;
\draw    (185,75) -- (185,7) ;
\draw [shift={(185,5)}, rotate = 90] [color={rgb, 255:red, 0; green, 0; blue, 0 }  ][line width=0.75]    (10.93,-3.29) .. controls (6.95,-1.4) and (3.31,-0.3) .. (0,0) .. controls (3.31,0.3) and (6.95,1.4) .. (10.93,3.29)   ;
\draw  [dash pattern={on 5.63pt off 4.5pt}][line width=1.5]  (185,25) -- (215,25) -- (215,75) -- (185,75) -- cycle ;
\draw   (185,86) .. controls (185,90.67) and (187.33,93) .. (192,93) -- (225,93) .. controls (231.67,93) and (235,95.33) .. (235,100) .. controls (235,95.33) and (238.33,93) .. (245,93)(242,93) -- (278,93) .. controls (282.67,93) and (285,90.67) .. (285,86) ;
\draw   (285,86) .. controls (285,90.67) and (287.33,93) .. (292,93) -- (325,93) .. controls (331.67,93) and (335,95.33) .. (335,100) .. controls (335,95.33) and (338.33,93) .. (345,93)(342,93) -- (378,93) .. controls (382.67,93) and (385,90.67) .. (385,86) ;
\draw  [pattern=_0up0rcxv1,pattern size=6pt,pattern thickness=0.75pt,pattern radius=0pt, pattern color={rgb, 255:red, 0; green, 0; blue, 0}] (275,45) -- (285,45) -- (285,75) -- (275,75) -- cycle ;
\draw  [pattern=_bbr4bjnyd,pattern size=6pt,pattern thickness=0.75pt,pattern radius=0pt, pattern color={rgb, 255:red, 0; green, 0; blue, 0}] (365,50) -- (385,50) -- (385,75) -- (365,75) -- cycle ;
\draw   (456,87) .. controls (456.07,91.67) and (458.44,93.96) .. (463.11,93.88) -- (476.62,93.66) .. controls (483.28,93.55) and (486.65,95.83) .. (486.73,100.5) .. controls (486.65,95.83) and (489.94,93.44) .. (496.61,93.33)(493.61,93.38) -- (510.12,93.11) .. controls (514.78,93.04) and (517.07,90.67) .. (517,86) ;
\draw  [fill={rgb, 255:red, 155; green, 155; blue, 155 }  ,fill opacity=1 ] (305,50) -- (315,50) -- (315,75) -- (305,75) -- cycle ;
\draw  [fill={rgb, 255:red, 155; green, 155; blue, 155 }  ,fill opacity=1 ] (285,45) -- (305,45) -- (305,75) -- (285,75) -- cycle ;
\draw  [fill={rgb, 255:red, 155; green, 155; blue, 155 }  ,fill opacity=1 ] (385,50) -- (395,50) -- (395,75) -- (385,75) -- cycle ;
\draw   (315,50) -- (335,50) -- (335,75) -- (315,75) -- cycle ;
\draw   (415,55) -- (425,55) -- (425,75) -- (415,75) -- cycle ;
\draw  [fill={rgb, 255:red, 155; green, 155; blue, 155 }  ,fill opacity=1 ] (455,65) -- (485,65) -- (485,75) -- (455,75) -- cycle ;
\draw    (280,40) .. controls (288.21,11.43) and (301.91,11.01) .. (309.65,43.49) ;
\draw [shift={(310,45)}, rotate = 257.29] [color={rgb, 255:red, 0; green, 0; blue, 0 }  ][line width=0.75]    (10.93,-3.29) .. controls (6.95,-1.4) and (3.31,-0.3) .. (0,0) .. controls (3.31,0.3) and (6.95,1.4) .. (10.93,3.29)   ;
\draw  [fill={rgb, 255:red, 155; green, 155; blue, 155 }  ,fill opacity=1 ] (395,55) -- (415,55) -- (415,75) -- (395,75) -- cycle ;
\draw    (375,45) .. controls (383.21,16.43) and (396.91,16.01) .. (404.65,48.49) ;
\draw [shift={(405,50)}, rotate = 257.29] [color={rgb, 255:red, 0; green, 0; blue, 0 }  ][line width=0.75]    (10.93,-3.29) .. controls (6.95,-1.4) and (3.31,-0.3) .. (0,0) .. controls (3.31,0.3) and (6.95,1.4) .. (10.93,3.29)   ;

\draw (182,79.4) node [anchor=north west][inner sep=0.75pt]  [font=\tiny]  {$0$};
\draw (212,79.4) node [anchor=north west][inner sep=0.75pt]  [font=\tiny]  {$1$};
\draw (242,79.4) node [anchor=north west][inner sep=0.75pt]  [font=\tiny]  {$2$};
\draw (272,79.4) node [anchor=north west][inner sep=0.75pt]  [font=\tiny]  {$3$};
\draw (302,79.4) node [anchor=north west][inner sep=0.75pt]  [font=\tiny]  {$4$};
\draw (332,79.4) node [anchor=north west][inner sep=0.75pt]  [font=\tiny]  {$5$};
\draw (362,79.4) node [anchor=north west][inner sep=0.75pt]  [font=\tiny]  {$6$};
\draw (392,79.4) node [anchor=north west][inner sep=0.75pt]  [font=\tiny]  {$7$};
\draw (403,76.4) node [anchor=north west][inner sep=0.75pt]  [font=\footnotesize]  {$\cdots $};
\draw (514,78.4) node [anchor=north west][inner sep=0.75pt]  [font=\tiny]  {$m$};
\draw (173,51) node [anchor=north west][inner sep=0.75pt]  [font=\scriptsize,rotate=-270] [align=left] {utility};
\draw (479,77) node [anchor=north west][inner sep=0.75pt]  [font=\scriptsize] [align=left] {goods};
\draw (231,101.4) node [anchor=north west][inner sep=0.75pt]  [font=\scriptsize]  {$I_{1}^{i}$};
\draw (331,100.4) node [anchor=north west][inner sep=0.75pt]  [font=\scriptsize]  {$I_{1}^{i}$};
\draw (483,100.4) node [anchor=north west][inner sep=0.75pt]  [font=\scriptsize]  {$I_{k_{i}}^{i}$};

\end{tikzpicture}
    \caption{Construction of the fractional allocation \(T'\) (shaded in grey) that is SD dominated by \(T\). }
    \label{fig:6}
\end{figure}
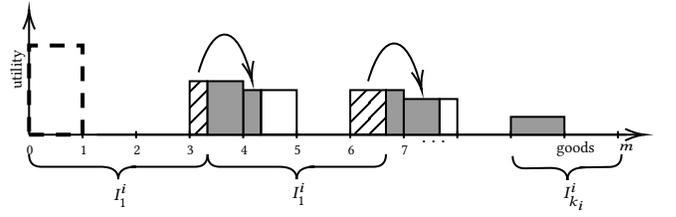
\end{proof}

With the help of this characterization, we now construct an \textit{allocation graph} \(G_g\) of goods. Given a fair allocation instance \(\mathcal{I}=\langle A, B, \Pi, \mathcal{F} \rangle\), we construct a bipartite graph - the allocation graph \(G_g=(S\cup B,E)\) as follows:

\renewcommand{\labelitemi}{-}
\begin{itemize}
    \item The set of goods \(B\) forms one part of \(G_g\) with goods interpreted as vertices. 
    \item For every agent \(a_i\in A\), and every \(\ell = 1,2,\cdots, m_i=\lceil m\alpha_i\rceil -1\), create vertex \(s_{i,\ell}\) in \(S\). We call these the \(m_i\) many \emph{slots} of agent \(a_i\).
    \item From each slot \(s_{i,\ell}\), draw edges to every good \(b\) for which \( \pi_i^{-1}(b) \le \left\lfloor\frac{\ell}{\alpha_i}\right\rfloor +1 \). That is, \((s_{i,\ell},b) \in E \iff \pi_i^{-1}(b) \le \left\lfloor\frac{\ell}{\alpha_i}\right\rfloor +1 \) 
\end{itemize}

The allocation graph \(G_g\) exhibits several interesting properties:

Firstly, we have 
\begin{proposition}\label{perfect-matchings-goods-wsdprop1}
An \(S\)-perfect matching in \(G_g\), i.e a matching that saturates all the goods, satisfies Conditions~\ref{eq:1} and \ref{eq:2} and this corresponds to a \wsdpropone{} allocation of goods\footnote{Note that this allocation could be a partial allocation. However, any completion of a \wsdpropone{} partial allocation remains \wsdpropone{}}
\end{proposition}
Moreover, in the interval set \(I^i=\{I^i_1, I^i_2,\ldots,I^i_{k_i}\}\) of an agent \(a_i\), if non-zero fraction of a good \(b\) lies in the interval \(I^i_\ell = \left[ \frac{\ell-1}{\alpha_i},\frac{\ell}{\alpha_i} \right]\), then \(\pi_i^-1(b)-1\le \left\lfloor\frac{\ell}{\alpha_i}\right\rfloor\). Thus slot \(s_{i,\ell}\) has an edge to good \(b\) in \(G_g\). This is formally stated in the following proposition.

\begin{proposition}\label{prop:slot-interval-good}
    In the allocation graph \(G_g\) of a good allocation instance \(\mathcal{I}\), each slot \(s_{i,\ell}\) of each agent \(a_i\in A\), has edges to every good with a non-zero portion in the interval \(I^i_\ell\) in the interval set of \(a_i\).
\end{proposition}

We also observe the following property regarding edge relationships in \(G_g\):
\begin{proposition}\label{prop:superset_goods}
     Let \(S_i = \{s_{i,\ell} \mid \ell \in [m_i]\}\) represent the set of slots in \(S\) that belong to an agent \(a_i\in A\). For any two slots \(s_{i,p}\) and \(s_{i,q}\) in \(S_i\), such that \(p<q\), the neighbourhood of \(s_{i,q}\) in \(B\) contains the neighbourhood of \(s_{i,p}\) in \(B\). That is, \(N(s_{i,p}) \subseteq N(s_{i,q})\). Therefore, we have:
    \[
         N(s_{i,1})\cup N(s_{i,2})\cup\cdots\cup N(s_{i,k})\subseteq N(s_i,k)
    \]
\end{proposition}

We also note the following observation about agents with different entitlements:
\begin{proposition}\label{prop:slot_compare}
    For any two agents \(a_i\) and \(a_k\) s.t \(\alpha_i \ge \alpha_k\), agent \(a_i\) has at least as many slots as agent \(a_k\) (because of Condition~\ref{eq:1}) and the size of the neighbourhood of \(\ell\)\textsuperscript{th} slot of agent \(a_i\) is smaller than or equal to the size of the neighbourhood of \(\ell\)\textsuperscript{th} slot of agent \(a_k\) . That is, \[\forall 1\le\ell\le\lceil m\alpha_k\rceil-1,\ |N(s_{i,\ell})|\le|N(s_{k,\ell})|\]
\end{proposition}

We are now ready to show the existence of \wsdpropone{} allocations by showing that an \(S\)-perfect matching always exists in \(G_g\), using the Hall's marriage theorem \cite{Hall1987}.

\begin{theorem}\label{thm:WPROP1-goods-exists}
    For any fair allocation instance of goods \(\mathcal{I}=\langle A,B,\Pi,\mathcal{F} \rangle\), there always exists a \wsdpropone{} allocation.
\end{theorem}
\begin{proof}
    Consider a subset \(T\) of slots in \(S\), and we aim to show \(|T| \leq |N(T)|\). Let \(a'_1, a'_2, \ldots, a'_k\) represent the agents with slots in set \(S\), and let \(\alpha'_1, \alpha'_2, \ldots, \alpha'_k\) denote their respective entitlements. Additionally, let \(\alpha_1, \alpha_2, \ldots, \alpha_k\) represent the entitlements of the top \(k\) most entitled agents in \(A\). Without loss of generality, we assume \(\alpha'_1 \geq \alpha'_2 \geq \ldots \geq \alpha'_k\) and \(\alpha_1 \geq \alpha_2 \geq \ldots \geq \alpha_k\).
    
    For each agent \(a'_p\) whose slot is present in \(T\), let \(s_{i_p}\) denote the highest index slot of agent \(a'_p\) in set \(T\). From Proposition~\ref{prop:superset_goods}, we can safely assume that all the \(i_p\) slots from \(s_1\) to \(s_{i_p}\) are included in set \(T\). Therefore the size of \(T\), \(|T| \le \sum_{p=1}^k i_p\).

    \begin{align*}
        |N(T)| &\ge \max(|N(s_{i_1})|,|N(s_{i_2})|,\ldots, |N(s_{i_k})|)\\
        &=\max\left(\left \lfloor \frac{i_1}{\alpha'_1} \right \rfloor, \left \lfloor \frac{i_2}{\alpha'_2} \right \rfloor, \ldots, \left \lfloor \frac{i_k}{\alpha'_k} \right \rfloor\right)+1&(\because \text{Construction of\ } G_g) \\
        &\ge  \max\left(\left \lfloor \frac{i_1}{\alpha_1} \right \rfloor, \left \lfloor \frac{i_2}{\alpha_2} \right \rfloor, \ldots, \left \lfloor \frac{i_k}{\alpha_k} \right \rfloor\right)+1&\text{($\because$ Proposition~\ref{prop:slot_compare})}\\
    \end{align*}
    
    As \(\sum_{p=1}^k \alpha_p \le 1\), there is at least one \(\alpha_q \in \{\alpha_1,\alpha_2,\cdots,\alpha_k\}\) s.t \(\alpha_q \le \frac{i_q}{i_1 + i_2 + \cdots + i_k}\). Therefore,
    
    \begin{align*}
     \max\left(\left \lfloor \frac{i_1}{\alpha_1} \right \rfloor, \left \lfloor \frac{i_2}{\alpha_2} \right \rfloor, \ldots, \left \lfloor \frac{i_k}{\alpha_k} \right \rfloor\right)+1 &\ge \left \lfloor \frac{i_q}{\alpha_q}\right \rfloor + 1\\
     &\ge \sum_{p=1}^k i_p + 1 &\text{$\because$ \(\alpha_q \le \frac{i_q}{\sum_{p=1}^k i_p}\) }\\
     &\ge |T|
    \end{align*}
    
    Therefore, because of Theorem~\ref{thm:halls}, the allocation graph \(G_g\) always has an \(S\)-perfect matching.
\end{proof}

Thus, using the famous Hopcroft-Karp algorithm \cite{Hopcroft-Karp} to find perfect matchings, we can compute a \wsdpropone{} allocation in time \(\mathcal{O}(m+n)^{2.5}\)

\subsection{Optimizing Over All Allocations}\label{subsec:opt-goods}
Similar to the case of chores, we first extend the allocation graph \(G_g\) to \(G^+_g\) so that perfect matchings in \(G^+_g\) corresponds to \wsdpropone{} allocations and vice versa. 

\paragraph{Extending the goods allocation graph:} Consider the allocation graph \(G_g=(S\cup B,E)\) of an instance of goods allocation. For each agent \(a_i\), there are \(m_i=\lceil m\alpha_i \rceil-1\) many slots in \(S\). To construct the extended allocation graph \(G^+_g=(S'\cup B',E')\), create \(q=m-|S|\), that is \(q=m+n-\sum_{i\in[n]}\lceil m\alpha_i \rceil\) many additional \emph{spare} slots \(s'_{i,m_i+1},s'_{i,m_i+2},\cdots,s'_{i,m_i+q}\) for \emph{every} agent \(a_i\in A\). Now the total number of slots is \(|S'|=|S|+nq\). To balance this bipartite graph, create \(t=|S'|-m\) many additional \emph{dummy} goods \(b'_1,b'_2,\cdots,b'_t\) in \(B'\). Draw edges from each spare slot to every good (both real and dummy) to extend \(E\) to \(E'\). 

Any perfect matching in \(G^+_g\) corresponds to an \(S\)-perfect matching in \(G_g\) and hence a \wsdpropone{} allocation. Similarly, given an \wsdpropone{} allocation \(X\), we know that every bundle in \(X\) satisfies conditions~\ref{eq:1} and \ref{eq:2} and therefore, for each agent \(a_i\), the \(m_i\) many slots can be matched according to \(X\). Such a matching can easily be extended to a perfect matching as every spare slot has edges to all the remaining goods.

We now optimize over all allocations with the help of the matching polytope given by Equation~\ref{eqn:LP}. Suppose a central agency, such as a government, wants to is proportionally allocate goods to agents living in different locations. For a given agent-good pair \(a_i,b_j\), there is an associated transportation cost denoted by \(u_i(b_j)\). We can find in polynomial time, a \wsdpropone{} allocation that minimizes the total transportation cost by minimizing the following objective function subject to the constraints given by Equations~\ref{eqn:LP}.
\[
    \sum_{i\in[n]}\sum_{b_j\in B} u_i(b_j) 
\]
Similarly, we can optimize for any linear objective function.

\subsection{Best of Both Worlds Fairness}\label{subsec:bobw-goods}
In \cite{bobw2023_1} and \cite{bobw2023_2}, the authors give a polynomial time algorithm to compute Ex-ante \wsdef{} and Ex-post \wsdpropone{} allocations of goods. The technique used mainly relies on the \emph{Probabilistic Serial} rule\cite{BOGOMOLNAIA2001295}, also called as the \emph{Eating Algorithm}. Here we give an alternate algorithm using the matching polytope of the extended allocation graph \(G^+_g\).

Consider the fractional allocation \(X\) where each agent \(a_i\in A\) receives \(\alpha_i\) fraction of every real good. For each pair of agents \(a_i\) and \(a_k\), under any two valuations \(v_i\in \mathscr{U}(\pi_i),v_k\in \mathscr{U}(\pi_k)\), we know that \(\frac{v_i(X_i)}{\alpha_i} = \frac{v_i(X_k)}{\alpha_k} = v_i(B)\) and hence \(X\) is a \wsdef{}  allocation. We first show that this fractional allocation can be realized as a fractional perfect matching in the extended allocation graph \(G^+_g\).

\begin{lemma}\label{lemma:fractionalmatching-goods}
    Given an instance \(\mathcal{I}=\langle A,B,\Pi,\mathcal{F} \rangle\) of goods allocation, there exists a fractional perfect matching in the extended allocation graph \(G^+_g\) that corresponds to a the \wsdef{} allocation where each agent \(a_i\in A\) receives \(\alpha_i\) fraction of every good.
\end{lemma}
\begin{proof}
    We construct first fractional matching that saturates all the \emph{real} (non dummy) goods and all the real (non spare) slots . Such a matching can always be extended to a fractional perfect matching by assigning the dummy goods in any manner, as all the dummy goods have edges to all the spare slots. 
    
    Consider the interval set \(I^i\) of an agent \(a_i\in A\). From Proposition~\ref{prop:slot-interval-good}, we know that slot \(s_{i,\ell}\) has edges to every good in the interval \(I^i_\ell\). With the help of this fact, we construct a fractional matching in \(G^+_g\) as follows:
    
    Let \(x_{i,\ell,b}\) denote the fraction of the edge \((s_{i,\ell},b)\) in the matching. let \(\delta_{b,\ell}\) denote the fraction of a good \(b\in B\) that is present in the interval \(I^i_\ell\). For every edge \((s_{i,\ell},b)\), we set \(x_{i,\ell,b}=\alpha_i\cdot\delta_{b,\ell}\). A slot \(s_{i,\ell}\) receives non zero fractions of the good from the interval \(I^i_\ell\). Each slot receives at most \(1\) unit of good because total good assigned for a slot \(s_{i,\ell}\) is : 
    \begin{align*}
    \sum_{b\in B}x_{i,\ell,b} = \alpha_i\sum_{b\in B}\delta_{b,\ell} \le \alpha_i\frac{1}{\alpha_j} = 1
    \end{align*}

    Since \(m_i = \lceil m\alpha_i\rceil -1\) and number of intervals is \(\lceil m\alpha_i\rceil\), the goods in the last interval could be unallocated. We allocate \(\alpha_i\) fraction of every good in this interval as well to one of the spare slots of agent \(a_i\) thus exhausting all the real goods.
    
    The fraction of a given real good \(b\) received by agent \(i\) across all the slots is:
    \begin{align*}
        \sum_{\ell=1}^{m_i} x_{i,\ell,b} =  \alpha_i\sum_{\ell=1}^{m_i}\delta_{b,\ell} = \alpha_i
    \end{align*}

    This gives us a matching \(M\) which saturates all the real goods and real slots. Since the graph \(G^+_g\) is a balanced bipartite graph, and as every dummy good has an edge to every spare slot, the matching \(M\) can be extended to a fractional perfect matching by dividing the dummy goods across the remaining spaces of all the spare slots in any arbitrary way. 
\end{proof}

Let us denote this fractional perfect matching as \(M^*\). Note that \(M^*\) lies inside the matching polytope of \(G^+_g\). We now decompose this fractional perfect matching into convex combination of integral perfect matchings with the help of Birkhoff’s decomposition as given in Theorem~\ref{thm:bvn}. We design Algorithm~\ref{alg:BoBW-goods}: The Uniform Lottery Algorithm, which gives an ex-ante \wsdef{} and ex-post \wsdpropone{} allocation of goods using only the ordinal valuations.

\begin{algorithm}
\caption{\textsc{Uniform Lottery Algorithm} for goods}\label{alg:BoBW-goods}
\begin{algorithmic}[1]
    \Require A good allocation instance \(\mathcal{I}=\langle{A},B,\Pi,\mathcal{F}\rangle\), where \(|A|=n\) and \(|B|=m\).
    \Ensure A fractional \wsdef{} allocation \(X=\sum_{k=1}^q~\lambda_k~X_k\) where each \(X_k\) represents a deterministic \wsdpropone{} allocation and \(q\in \mathcal{O}(m^c)\).
    \State \(G^+_g \gets\) extended allocation graph of \(\mathcal{I}\)
    \State \(Y \gets \) fractional perfect matching in \(G^+_g\) where each agent \(a_i\in A\) gets \(\alpha_i\) fraction of every real good \Comment{(As in Lemma~\ref{lemma:fractionalmatching-goods})}
    \State Invoke Theorem~\ref{thm:bvn} to compute a decomposition \(Y=\sum_{k=1}^q~\lambda_k Y_k\) where \(k\le(m+n)^2-(m+n)-2\)
    \State Convert \(Y=\sum_{k=1}^q\lambda_k{Y_k}\) to \(X=\sum_{k=1}^q\lambda_k{X_k}\) where all the dummy goods are ignored.\\
    \Return Allocation \(X\) and its decomposition \(\sum_{k=1}^q\lambda_k{X_k}\)
\end{algorithmic}
\end{algorithm}

\begin{theorem}
    The randomized allocation implemented by Algorithm~\ref{alg:BoBW-goods} is ex-ante \wsdef{} and ex-post \wsdpropone{}
\end{theorem}
\begin{proof}
    Algorithm~\ref{alg:BoBW-goods} returns an allocation \(X\) and its decomposition \(\sum_{k=1}^q\lambda_k{X_k}\). From Lemma~\ref{lemma:fractionalmatching-goods}, we know that the allocation \(X\) returned by the algorithm is \wsdef{}. Each of the \(X_k\)s in the decomposition is a \(B\)-perfect matching in the allocation graph \(G_g\).Therefore, from Proposition~\ref{perfect-matchings-goods-wsdprop1}, each \(X_k\) is \wsdpropone{}.
\end{proof}

\subsection{Sequencible Allocations via Rank-Maximal Perfect Matchings}\label{subsec:seq-goods}
In this section, we use the Lemma~\ref{lemma:rank-maximal-perfect} to show that a rank-maximal \(S\)-perfect matching in the  allocation graph \(G_g\) gives a sequencible \wsdpropone{} allocation.

\begin{theorem}\label{thm:seq-goods}
    There always exist a \wsdpropone{}+\seq{} allocation of goods.
\end{theorem}
\begin{proof}
Let \(G_g=(S,B,E)\) be the allocation graph of an instance \(\mathcal{I}\). For each slot \(s_{i,\ell}\), we first rank the goods from \(1\) to \(m\) as \(rank(s_{i,\ell},b)=\pi_i(b)\) for all \(b\in B\). For any two slots \(s_{i,p}\) and \(s_{i,q}\) of an agent \(a_i\), if \(p>q\), then \(N(q)\subseteq N(p)\). This is because \(G_g\) satisfies Condition~\ref{eq:2}. Therefore, given a matching \(M\), if \(rank(s_{i,p},M(s_{i,p})<rank(s_{i,q},M(s_{i,q}))\), then we can interchange \(M(s_{i,p})\) and \(M(s_{i,q})\) without altering the signature of the matching. Thus, given a rank-maximal \(S\)-perfect matching \(M\), we can assume w.o.l.g that for any agent \(a_i\in A\), and \(p,q\le m_i\), if \(p>q\) then \(rank(s_{i,p},M(s_{i,p})>rank(s_{i,q},M(s_{i,q}))\).

Given a rank-maximal \(S\)-perfect matching \(M\) in \(G_g\), from Lemma~\ref{lemma:rank-maximal-perfect} we obtain a sequence \(\sigma(S)\) of slots. To construct a sequence of agents, replace each \(s_{i,\ell}\) with the corresponding agent \(a_i\). Therefore, a rank-maximal \(S\)-perfect matching in \(G_g\) gives a \wsdpropone{}+\seq{} allocation. 
\end{proof}

Therefore, using the algorithm to find rank-maximal perfect matchings \cite{michail2007reducing,irving2003greedy}, we can compute a \wsdpropone{}+\seq{} allocation of goods in time \(\mathcal{O}((m+n)^{3.5})\).

\section{Examples}\label{sec:examples}
In this section we give examples and instances for some of our claims.
\subsection{Non-Existence of \wsdpropx{} Allocations}\label{subsec:no_wpropx}
Consider an ordinal instance with three agents \(a_1,a_2,\) and \(a_3\) with equal entitlements and three goods \(b_1,b_2,b_3\). Let the agents have identical ranking \((b_1\succ b_2 \succ b_3)\) over the goods. For this instance, any \wsdpropx{} allocation must allocate one good per agent. Because otherwise, there would be an agent who gets no good and the allocation is not \wpropx{} under the valuation \((1,1,\epsilon)\).

Now consider an allocation where agents gets one item each. Without loss of generality, let agent \(a_1\) get the least valuable good \(b_3\). This allocation, however, is not \wpropx{} under the valuation \((1,\epsilon,0)\) for agent \(a_3\). Therefore, this instance admits no \wsdpropx{} allocation. It is evident that a similar example can be constructed even for the case of chores. 

\subsection{Impossibility of Pareto Optimal Allocations of Goods}\label{subsec:impossibility-goods}
We construct an instance \(\mathcal{I}=\langle A,B,\Pi,\mathcal{F} \rangle\) such that no \wsdpropone{} allocation is \po{} under all \(\Pi\)-respecting valuations. Consider two agents \(A=\{a_1,a_2\}\) with three goods \(B=\{b_1,b_2,b_3\}\). Let both the agents have identical ordinal preference: \(\pi_1(1)=\pi_2(1)=b_1, \pi_1(2)=\pi_2(2)=b_2\) and \(\pi_1(3)=\pi_2(3)=b_3\). Let both the agents have equal entitlements: \(\mathcal{F}=\{0.5,0.5\}\).

From the characterization of \wsdpropone{} allocations, we know that a \wsdpropone{} allocation must satisfy conditions \ref{eq:1} and \ref{eq:2} of Lemma~\ref{lemma:WPROP1-goods}. Applying condition \ref{eq:1}, each agent must receive at least one good in any \wsdpropone{} allocation. Suppose \emph{w.l.o.g} agent \(a_1\) receives good \(b_1\). We then set \(v_1=(1,1,1)\). Agent \(a_2\) must receive at least one of \(b_2,b_3\). Suppose \emph{w.l.o.g} agent \(a_2\) receives \(b_2\), we then set \(v_2=(1,\epsilon,\epsilon)\) where \(\epsilon<1\). The good \(b_3\) can be assigned to any agent (Refer to Table~\ref{tab:counter-example-PO-goods}). We see that this \wsdpropone{} allocation cannot be Pareto optimal as trading the good \(b_1\) with \(b_2\) gives a Pareto improvement.  

Note that in the above example, agent \(a_1\) values all the items equally. However, this equal valuation is not a necessary requirement to show the incompatibility of \wsdpropone{} with cardinal \po{}. Specifically, even if we impose a condition that items with higher rankings must have strictly higher values, we can still create analogous instances. For example, consider an instance with two agents and five goods. Let the agents have identical entitlements and ordinal rankings. It is clear from Condition~\ref{eq:1} that any \wsdpropone{} allocation must give at least two items to each agent. Now, we always can set a consistent cardinal valuation (with unequal values), that incentives trading the \(1^{\text{st}}\) item with two items of the other agent, giving a Pareto improvement.    

\begin{table}
\centering
    \begin{tabular}{lccc}
        \toprule
        \multicolumn{1}{l}{Goods} & $b_1$ & $b_2$ & $b_3$ \\ 
        \midrule
        Agent $a_1$ & \fcolorbox{gray}{gray!20}{$1$} & $1$ & $1$ \\  
        Agent $a_2$ & $1$ & \fcolorbox{gray}{gray!20}{$\epsilon$} & $\epsilon$ \\    
        \bottomrule
    \end{tabular}
    \caption{Exchanging \(b_1\) with \(b_2\) gives a Pareto dominating allocation.}
    \label{tab:counter-example-PO-goods}
\end{table}

\subsection{Impossibility of Pareto Optimal Allocation of Chores}\label{subsec:impossibility-chores}
We construct an instance \(\mathcal{I}=\langle A,B,\Pi,\mathcal{F} \rangle\) such that no \wsdpropone{} allocation is \po{} under all \(\Pi\)-respecting valuations. Consider two agents \(A=\{a_1,a_2\}\) with three chores \(B=\{b_1,b_2,b_3\}\). Let both the agents have identical ordinal preference: \(\pi_1(1)=\pi_2(1)=b_1, \pi_1(2)=\pi_2(2)=b_2\) and \(\pi_1(3)=\pi_2(3)=b_3\). Let both the agents have equal entitlements: \(\mathcal{F}=\{0.5,0.5\}\).

From the characterization of \wsdpropone{} allocations, we know that a \wsdpropone{} allocation must satisfy conditions \ref{eq:3} and \ref{eq:4} of Lemma~\ref{lemma:WPROP1-chores}. Applying condition \ref{eq:3}, we see that no agent should receive more than two chores. That is, each agent must receive at least one chore in any \wsdpropone{} allocation. Suppose \emph{w.l.o.g} agent \(a_1\) receives chore \(b_3\). We then set \(v_1=(1,1,1)\). Agent \(a_2\) must receive at least one of \(b_1,b_2\). Suppose \emph{w.l.o.g} agent \(a_2\) receives \(b_2\), we then set \(v_2=(1,1,\epsilon)\) where \(\epsilon<1\). The chore \(b_1\) can be assigned to any agent (Refer to Table~\ref{tab:counter-example-PO-chore}). We see that this \wsdpropone{} allocation cannot be Pareto optimal as trading the chores \(b_2\) with \(b_3\) gives a Pareto improvement. In Appendix section~\ref{subsec:impossibility-goods} we give a similar example for the case of goods. 

Furthermore, there exist instances where given the cardinal valuations, there does not exist a Pareto optimal allocation that is \wsdpropone{} for the underlying ordinal instance. For example, consider an instance where an agent \(a_i\) values more than \(\alpha_i.m\) many chores at value zero and no other agent values those chores at value zero. Therefore, under any Pareto optimal allocation, agent \(a_i\) gets more than or equal to \(\alpha_i.m\) many chores. However, such an allocation cannot be \wsdpropone{} for the underlying ordinal instance as it violates the Condition~\ref{eq:3}. 

\begin{table}
\centering
    \begin{tabular}{lccc}
        \toprule
        \multicolumn{1}{l}{Chores} & $b_1$ & $b_2$ & $b_3$ \\ 
        \midrule
        Agent $a_1$ & $1$ & $1$ & \fcolorbox{gray}{gray!20}{$1$} \\  
        Agent $a_2$ & $1$ & \fcolorbox{gray}{gray!20}{$1$} & $\epsilon$ \\    
        \bottomrule
    \end{tabular}
    \caption{Exchanging \(b_3\) with \(b_2\) gives a Pareto dominating allocation.}
    \label{tab:counter-example-PO-chore}
\end{table}

\end{appendices}

\end{document}